\documentclass[
aip,
superscriptaddress, 
amsmath,amssymb,preprint]{revtex4-1}

 \usepackage{graphicx}
 \graphicspath{{pics/}} 
\usepackage{dcolumn}
\usepackage{bm}

\usepackage[utf8]{inputenc}
\usepackage[T1]{fontenc}
\usepackage{newtxmath} 
\usepackage{etoolbox} 
\usepackage{mathtools} 

\usepackage[retainorgcmds]{IEEEtrantools} 

\usepackage{extarrows} 

\usepackage[singlelinecheck=false,justification=raggedright,bf]{caption}
\usepackage{subcaption}

\usepackage{xcolor}
\usepackage{tikz}

\usepackage{ytableau}
\ytableausetup{mathmode, centertableaux, boxsize=3mm}

\usepackage{cancel}

\usepackage{pbox}

\usepackage[pdftex,colorlinks]{hyperref} 
 \hypersetup{
     colorlinks,
     linkcolor={red!50!blue},
     citecolor={blue!50!green},
     urlcolor={blue!80!black}
   }

\usepackage[ntheorem]{empheq}
\usepackage[hyperref]{ntheorem}
\newtheorem{theorem}{Theorem}
\newtheorem{lemma}{Lemma}

\newcommand{\qed}{\hfill\tikz{\draw[draw=black,line width=0.6pt] (0,0)
    rectangle (2.8mm,2.8mm);}\bigskip}

\newenvironment{proof}[1][]{
  \noindent\emph{
    Proof\ifthenelse{\equal{#1}{}}{.}{ of #1.}
  }
  \hspace{0.3cm}}
{\hfill\qed\medskip}

\makeatletter
\newcommand{\vast}{\bBigg@{3}}
\newcommand{\Vast}{\bBigg@{4}}
\newcommand{\Vvast}{\bBigg@{5}}
\newcommand{\VAST}{\bBigg@{6}}
\makeatother

\usepackage[nameinlink]{cleveref} 

\makeatletter
\def\@email#1#2{
 \endgroup
 \patchcmd{\titleblock@produce}
  {\frontmatter@RRAPformat}
  {\frontmatter@RRAPformat{\produce@RRAP{*#1\href{mailto:#2}{#2}}}\frontmatter@RRAPformat}
  {}{}
}
\makeatother

\begin{document}

\newcommand{\RPic}[2][{}]{\hspace{-0.4mm}\pbox{\textwidth}{\includegraphics[#1]{{#2}}}\hspace{-0.4mm}}
\newcommand{\FPic}[2][{}]{\hspace{-0.27mm}\pbox{\textwidth}{\includegraphics[#1]{{#2}}}\hspace{-0.27mm}}
\newcommand{\ybox}[1]{\begin{ytableau} #1 \end{ytableau}}

\newcommand{\SUN}{\mathsf{SU}(N)}
\newcommand{\SU}[1]{\mathsf{SU}(#1)}
\newcommand{\GLN}{\mathsf{GL}(N)}
\newcommand{\MixedPow}[2]{V^{\otimes
    #1}\otimes\left(V^*\right)^{\otimes #2}}
\newcommand{\Pow}[1]{V^{\otimes #1}}

\preprint{LU-TP 22-56, MCnet-22-18}

\title{Wigner $6j$ symbols for ${\rm SU}(N)$: Symbols with at least two quark-lines}

\author{Judith Alcock-Zeilinger}
\affiliation{Erwin Schr\"odinger Institute for Mathematics and Physics, University of Vienna, Boltzmanngasse 9, A-1090 Wien, Austria}
\affiliation{Fachbereich Mathematik, Universit\"at T\"ubingen, Auf der Morgenstelle 10, 72076 T\"ubingen, Germany}

\author{Stefan Keppeler}
\affiliation{Fachbereich Mathematik, Universit\"at T\"ubingen, Auf der Morgenstelle 10, 72076 T\"ubingen, Germany}

\author{Simon Pl\"atzer}
\affiliation{Institute of Physics, NAWI Graz, University of Graz, Universit\"atsplatz 5, A-8010 Graz, Austria}
\affiliation{Particle Physics, Faculty of Physics, University of Vienna, Boltzmanngasse 5, A-1090 Wien, Austria}
\affiliation{Erwin Schr\"odinger Institute for Mathematics and Physics, University of Vienna, Boltzmanngasse 9, A-1090 Wien, Austria}

\author{Malin Sjodahl}
\affiliation{Department of Astronomy and Theoretical Physics, Lund University,
Box 43, 221 00 Lund, Sweden}
\affiliation{Erwin Schr\"odinger Institute for Mathematics and Physics, University of Vienna, Boltzmanngasse 9, A-1090 Wien, Austria}

\date{29 September 2022}

\begin{abstract}

  We study a class of $\SUN$ Wigner $6j$~symbols involving two
  fundamental representations, and derive explicit formulae for all
  $6j$~symbols in this class. Our formulae express the $6j$~symbols in
  terms of the dimensions of the involved representations, and they
  are thereby functions of $N$. We view these explicit formulae as a
  first step towards efficiently decomposing $\SUN$ color structures
  in terms of group invariants.
\end{abstract}

\maketitle

\section{Introduction}
\label{sec:intro}

A unique feature of the characteristic quantum numbers of the strong
force as described by Quantum Chromodynamics (QCD) is that they are
confined and not observable. In order to extract observable quantities
from QCD scattering amplitudes one has to average over external color
quantum numbers or, otherwise, project onto definite hadronic states.
In both cases the color quantum numbers enter calculations on a
similar footing as internal interfering quantum mechanical degrees of
freedom. This allows for the usage of color bases which leave out
information on the states within an irreducible representation (irrep)
of $\SU{3}$.

Despite this simplification, one of the challenges in multi-parton QCD
calculations is the accurate description of the large color
space. Often color-summed/-averaged so-called trace\cite{Paton:1969je,
  Berends:1987cv, Mangano:1987xk, Mangano:1988kk, Kosower:1988kh,
  Nagy:2007ty, Sjodahl:2009wx, Alwall:2011uj, Sjodahl:2014opa,
  Platzer:2012np,Platzer:2018pmd} and color-flow
bases\cite{'tHooft:1973jz,Kanaki:2000ms,Maltoni:2002mq,Platzer:2013fha,AngelesMartinez:2018cfz,DeAngelis:2020rvq,Platzer:2020lbr}
are used. These bases take advantage of the possibility to ignore the
internal structure of $\SU{3}$ irreps, but they also ignore the irreps
altogether, i.e.\ the basis vectors are in no correspondence to the
intermediate states in which a set of partons transforms. Moreover,
for a finite number $N$ of colors, trace and color-flow bases are
non-orthogonal and overcomplete, i.e., strictly speaking, they are not
even bases but only spanning sets. The size of theses spanning sets
grows roughly as a factorial in the number $n_{g+q\overline{q}}$ of
gluons and $q\overline{q}$-pairs \cite{Keppeler:2012ih}. Since these
spanning sets are non-orthogonal, this translates, in the worst case,
to a factorial square scaling, $(n_{g+q\overline{q}}!)^2$, for the
number of inner products that have to be calculated. In a full color
description of a scattering cross section this growth will be
prohibitive if one calculates all contributing terms and cannot use
additional information about the amplitudes or exploit Monte Carlo
methods to sample over color structures, as {\it e.g.} done in
\cite{DeAngelis:2020rvq}.

An ideal basis would be both orthogonal and minimal, allowing for the
smallest number of terms needed when expanding amplitudes and
correlation functions in color structures. 
These properties are
combined in multiplet
bases\cite{Kyrieleis:2005dt,Dokshitzer:2005ig,Sjodahl:2008fz,Beneke:2009rj,
  Keppeler:2012ih,
  Du:2015apa,Sjodahl:2015qoa,Keppeler:2013yla,Alcock-Zeilinger:2016bss,Alcock-Zeilinger:2016sxc,Alcock-Zeilinger:2016cva,Sjodahl:2018cca},
which use representation theory to iteratively group partons into
orthogonal states. So far, the use of multiplet bases has been rather
limited, in part likely due to the lack of explicit bases for many
partons, i.e.\ {\it the} situation in which they would really be
advantageous.

In this paper we suggest taking the usage of representation theory one
step further: Instead of \textit{explicitly} created bases, we advocate
using Wigner $6j$~coefficients (or $6j$~symbols or just $6j$s -- we
use all these terms interchangeably in this paper) for calculations in
color space.  This, however, assumes that the required $6j$
coefficients have been calculated and are readily available also for a
high number of partons. In
Refs.~\citenum{Sjodahl:2015qoa,Sjodahl:2018cca} explicit bases were
used to calculate $6j$s for a limited number of partons.  This allows
for a fast decomposition at use-time, i.e.\ when the $6j$s
(corresponding to the same bases) are used in actual computations of
amplitudes, but the factorial growth of the spanning set remains a
challenge at construction time, thereby effectively limiting the
number of involved particles to one or a few handfuls.
  
When decomposing color structures into multiplet bases using $6j$s, no
explicit bases are needed; i.e.\ all calculations are performed in
terms of $\SUN$ group invariant $6j$ coefficients, along with
dimensions of representations and Wigner $3j$~coefficients (or
$3j$~symbol, or $3j$ for short), which may be normalized to $1$.  This
poses the question if it should not also be possible to derive these
invariants in terms of themselves. More precisely: Can one derive a
consistent set of $6j$ coefficients only in terms of group invariants,
specifically the dimensions of representations?

In this paper we answer this question affirmatively when at least two
of the irreps involved in a $6j$~symbol are fundamental
representations, i.e.\ quark-lines, that do not share a common
vertex. In \Cref{thm:6j-closed-form-expression} we present explicit
formulae for the absolute values of all $6j$~symbols in this class. We
also explain how to iteratively fix and determine the signs of these
$6j$s, and for $N\leq3$ we prove (whereas for $N>3$ we conjecture)
that this procedure always leads to a consistent set of signs. In
particular, we have thus determined these $6j$~symbols in the
phenomenologically relevant case $N=3$. For $N=2$ the problem was
generally solved before, see e.g.\ Ref.~\citenum{Johansson:2016}. In
future work we will show how to determine the other $6j$s
required for a full color decomposition.  We note that other
approaches for calculating $SU(3)$ $6j$s in terms of Clebsch-Gordan
coefficients exist\cite{Alex_2011,Dytrych:2021qwe}, but stress that
our method exploits group invariants \textit{only}.

We view our results as a first step towards a complete reduction of
color space in terms of $\SUN$ group invariants, which has the
potential to significantly simplify fixed-order as well as all-order
calculations in color space, ranging from analytic approaches up to
Monte Carlo methods. The reduction of color space in terms of
invariants also has the potential to provide further insight into
other aspects such as the color structure of hadronization models
\cite{Gieseke:2018gff,Platzer:2022jny}.

This paper is organized as follows: \Cref{sec:computational} gives a
brief introduction to the birdtrack method for $\SUN$, illustrating
how $6j$~symbols appear and how they can be used to decompose more
general color structures. In \Cref{sec:6j-alpha-add-two-quark-lines}
we introduce the particular class of $6j$~symbols of interest to the
present paper, and in \Cref{sec:properties} we describe properties of
general $6j$~symbols, as well as of the class of symbols studied
here. \Cref{sec:closed-form-expressions} constitutes the main part of
this paper, wherein the closed form expressions of the $6j$s are
presented in \Cref{thm:6j-closed-form-expression}. The relevance of
these results, as well as future work complementing them, are
discussed in \Cref{sec:conclusions}.

  \section{Birdtrack methods for~\texorpdfstring{$\SUN$}{SU(N)} color space}
\label{sec:computational}

In this section we briefly outline how to utilize the birdtrack method
for decomposing group invariant (color) structures in terms of
dimensions, Wigner-$3j$ and Wigner-$6j$~symbols. For a full,
comprehensive introduction to the birdtrack formalism, readers are
referred to Ref.~\citenum{Cvitanovic:2008zz}. The hasty reader finds a
minimal introduction in Appendix~A of Ref.~\citenum{Keppeler:2012ih},
whereas a more pedagogical account can be found in
Ref.~\citenum{Keppeler:2017kwt}. Examples of birdtrack calculations
for QCD can be found in Ref. \citenum{Du:2015apa,Sjodahl:2015qoa}.

We start out with a reminder that implicit indices of states within a
representation are always summed over.
We therefore have, for an irrep $\alpha$,
\begin{equation}
  \label{eq:rep-loop}
  \raisebox{-0.5 cm}{\includegraphics[scale=0.5]{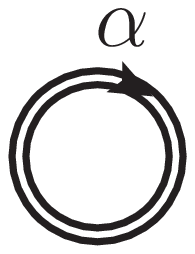}} 
  =d_\alpha\;,  
\end{equation} 
i.e.\ the sum of states within an irrep adds up to
the dimension  $d_\alpha$ of that irrep.

The second simplest color structure that may be encountered, which
also contains a sum, is the ``self energy'' diagram 
\begin{equation}
  \label{eq:self-energy}
  \raisebox{-0.8 cm}{\includegraphics[scale=0.5]{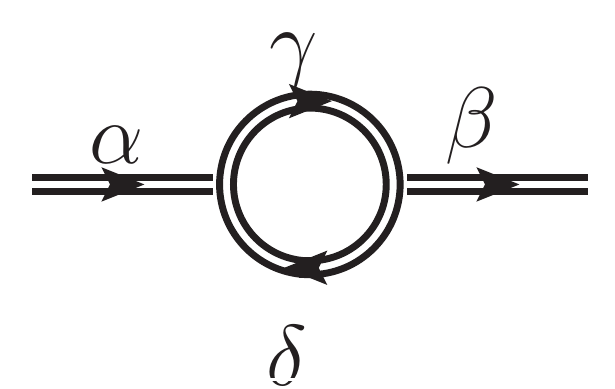}} 
  = \frac{\includegraphics[scale=0.5]{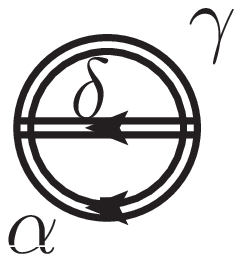}}{d_\alpha}
  \raisebox{-0.1 cm}{\includegraphics[scale=0.5]{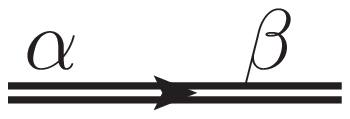}}
  \ ,
\end{equation} 
where the free line is to be understood as a Kronecker delta in
the representation indices $\alpha$ and $\beta$, and 
the normalization constant can be found by contracting indices
on both sides; as a consistency check, we have 
\begin{equation}
  \label{eq:self-energy-contracted}
  \raisebox{-1 cm}{\includegraphics[scale=0.5]{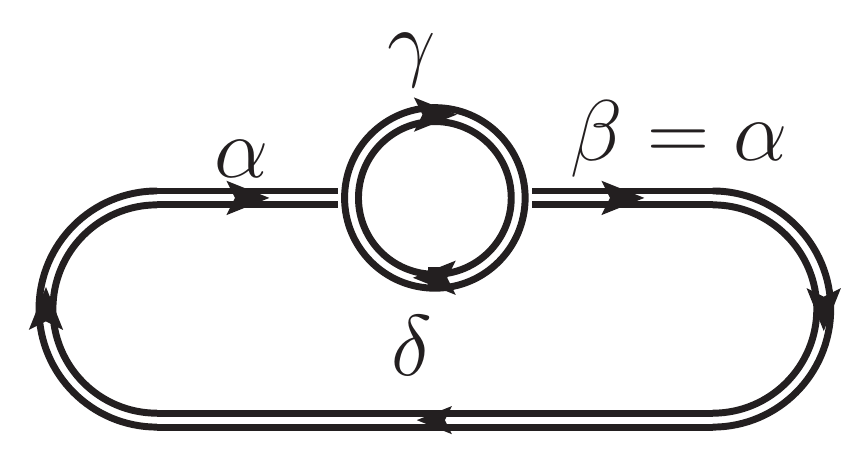}}
  =
  \frac{\includegraphics[scale=0.5]{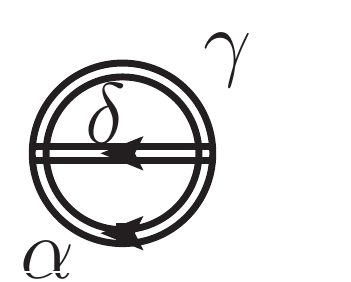}}{d_\alpha}
  d_\alpha
  =  \raisebox{-0.8cm}
  {\includegraphics[scale=0.5]{self-energy-contracted-2}}
  \ ,
\end{equation} 
where we used \Cref{eq:rep-loop}.
The result on the right hand side is known as a $3j$~symbol.
It is proportional to the magnitude of the vertex, and, depending
on the vertex normalization, it may thus assume different values.
We will keep the normalization of the $3j$~symbol arbitrary for
most of our derivations, although our final results are stated
in the normalization where all $3j$s are normalized to 1. This is 
in contrast to the standard QCD normalization for which, for
example
$\raisebox{-0.6cm}{\includegraphics[scale=0.5]{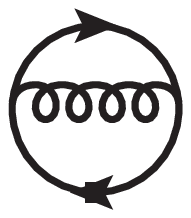}}= \frac{1}{2}(N^2-1)$,
for the generator normalization $\text{tr}[t^a
t^b]=\frac{1}{2}\delta^{a b}$.

After the self-energy, with the topology of a loop involving
two internal representations, the next structure to consider
is the vertex correction. Here, we also encounter the
Wigner-$6j$~symbols for the first time, as they act as normalization
constants when eliminating loops with three internal
lines,
  \begin{equation}
    \label{eq:vertex-correction}
    \RPic{GenRep-alpha-sigma-rhoSTAR--VCorr-delta-gamma-betaSTAR}
    \; =
    {\color{black!60}\sum_a}
    \;
    \frac{1}{\FPic{3j-sigmaSTAR-alpha-rho-VertexLabelsVa}}
    \;
    \underbrace{
      \RPic{6j-delta-rho-sigma--beta-gamma-alpha--CornerLabelsV3a}
    }_{\text{Wigner-$6j$}}
    \hspace{2mm}
    \RPic{GenRep-alpha-sigma-rhoSTAR--V--VLabela}
    \ ,
  \end{equation} 
  where the sum over all possible vertices $a$ collapses to only one
  term if the irreps $\alpha$, $\sigma$ and $\rho$ admit only
  one vertex. For the particular $6j$~symbols discussed in this paper,
  this is always the case. The triangular pictogram in
  \Cref{eq:vertex-correction} describes a Wigner-$6j$ symbol which, in
  the case of $\SU{2}$, is usually denoted in 2-line notation as
  \begin{equation}
    \label{eq:SU2-Wigner-6j}
    \RPic{6j-delta-rho-sigma--beta-gamma-alpha}
    \quad
    \xlongrightarrow{\SU{2}}
    \quad
    \begin{Bmatrix}
      j_{\beta} &  j_{\gamma} & j_{\alpha} \\
      j_{\delta} &  j_{\rho} & j_{\sigma}
    \end{Bmatrix}
    \ .
  \end{equation}

  For loops with more than three internal representations there is no
  similar simple expression. Instead, such color structures can be
  systematically reduced into loops with fewer internal lines by the
  application of the completeness relation
  \begin{equation}
    \label{eq:completeness-relation}
    \RPic{GenRep-beta-gamma}
    \; =
    \sum_{\delta}
    \frac{d_{\delta}}{\RPic{3j-gammaSTAR-beta-delta}}
    \;
    \RPic{GenRep-gammaVbeta-delta-ProjOps}
    \ .
  \end{equation} 
  Applying this to a loop with more internal irreps, the color
  structure can be rewritten in terms of a shorter loop and a sum 
  of vertex corrections, which may be removed using \Cref{eq:vertex-correction}.
  Schematically, we have
  \begin{equation}\label{eq:LoopContraction}
\raisebox{-0.45\height}{
	\includegraphics[scale=0.5]{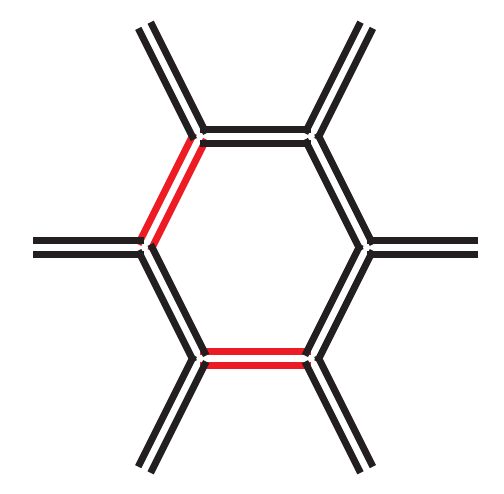}
}
\xlongequal{\eqref{eq:completeness-relation}}
\sum_{\alpha}{
\frac{
	d_\alpha
}{
	\hspace{-2mm}
	\raisebox{-0.45\height}{
		\includegraphics[scale=0.3]{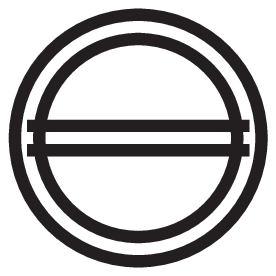}
	}
	\hspace{-2mm}
}
\raisebox{-0.45\height}{
	\includegraphics[scale=0.5]{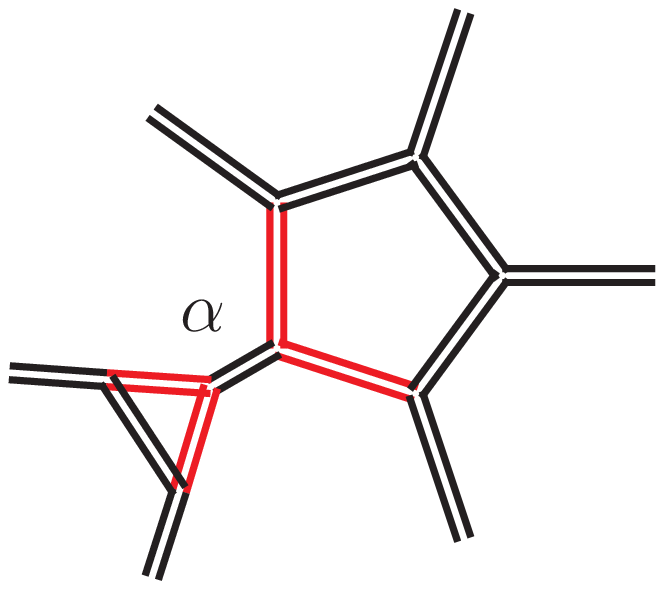}
}
}
\xlongequal{\eqref{eq:vertex-correction}}
\sum_{\alpha}{
\frac{
	d_\alpha
}{
	\hspace{-2mm}
	\raisebox{-0.45\height}{
		\includegraphics[scale=0.3]{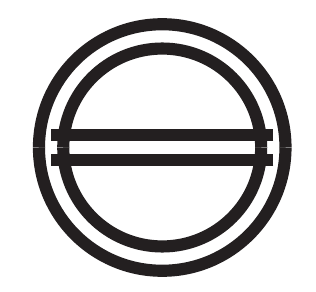}
	}
	\hspace{-2mm}
}
\frac{
	\hspace{-1mm}
	\raisebox{-0.2\height}{
		\includegraphics[scale=0.4]{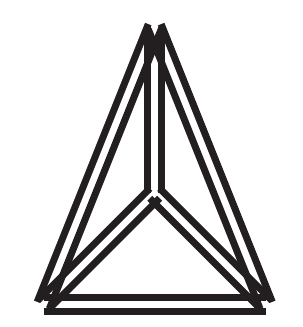}
	}
}{
	\hspace{-2mm}
	\raisebox{-0.45\height}{
		\includegraphics[scale=0.3]{Wig3jNoRepLabels}
	}
	\hspace{-2mm}
}
\raisebox{-0.45\height}{
	\includegraphics[scale=0.5]{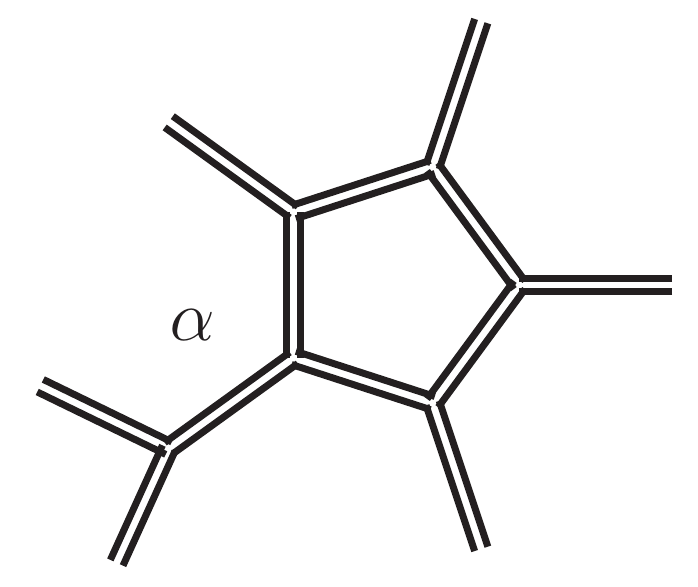}
}
},
\end{equation} 
  which can be fully reduced to $3j$s, $6j$s and dimensions by applying the completeness
  relation two more times.
    In a similar fashion, loops with yet more internal representations
  can be reduced back to expressions involving dimensions and $3j$~and 
  $6j$~coefficients.

  For this reason, to decompose an arbitrary color structure
  into group invariants, it is in principle enough to know the dimensions
  which may be calculated using standard methods (see
  e.g.\ \cite{Cvitanovic:2008zz,Fulton:1997,Sagan:2000}, also
  summarized in~\Cref{app:dimensions}),
  the $3j$~coefficients (which we normalize via the vertices to 1)
  and the $6j$~coefficients, a class of which will be derived here.

  \section{Wigner-\texorpdfstring{$6j$}{6j}~symbols with two opposing
    quark-lines}\label{sec:6j-alpha-add-two-quark-lines}

  In this work we focus on $6j$~symbols with (at least) two
  quark-lines on opposite edges,
\begin{equation}
  \label{eq:6j-Symbol}
  \FPic{6j-Mij-green-Mi--blueSTAR-Mj-alphaSTAR}
  \ ,
\end{equation} 
where the single lines are understood to be in the fundamental
representation, and $\alpha$, $M_i$, $M_j$ and $\bm{M}^{ij}$ are
irreps which can be thought about as Young diagrams. For the main part
of this paper, we will assume that none of the irreps labeled by
$\alpha$, $M_i$, $M_j$ and $\bm{M}^{ij}$ is the fundamental
representation, as this allows us to ignore irrep ordering in vertices
(this is explained in detail in
\Cref{app:properties-vertices}). However, \Cref{app:special} discusses
a few special cases where some of these irreps are indeed the
fundamental representation corresponding to $\ydiagram{1}$.

The Young diagrams $\alpha$, $M_i$, $M_j$ and $\bm{M}^{ij}$ used in
the construction of the $6j$~symbol given in \Cref{eq:6j-Symbol} are
related to each other as follows: 

\begin{itemize}
\item We begin by fixing a Young diagram $\alpha$.
\item Thereafter, we add a single box to $\alpha$ in row $i$
  (resp. $j$) in order to obtain $M_i$ (resp. $M_j$). Note that, in
  general, we cannot add a box to every row of $\alpha$ since in some
  cases the result would not be a Young diagram.
\item Lastly, $\bm{M}^{ij}$ is the diagram obtained from $\alpha$ by
  adding two boxes, first one in row $i$ and then one in row $j$. If
  $M_i$ and $M_j$ both exist, then the final result of adding two
  boxes is commutative, i.e.\
\begin{equation}
  \label{eq:Mij=Mji}
  \bm{M}^{ij} = \bm{M}^{ji}
  \ .
\end{equation} 
\end{itemize}
Examples for the construction of $M_i$, $M_j$ and $\bm{M}^{ij}$ from a
fixed diagram $\alpha$ are given in
\Cref{fig:Example-alpha-Mi-Mj-Mij}.

\begin{figure}[tb]
  \hrule
\centering
\begin{subfigure}[t]{0.49\textwidth}
  \centering
  \FPic{YDAncestry-3-1Pgreen1-Pbluehatched1}
  \caption{For $i=2$ and $j=3$, we obtain $M_2$, $M_3$ and
$\bm{M}^{23}$ from $\alpha$ by adding a green (shaded) box in row $2$
and a blue (hatched) box in row $3$.}
  \label{fig:Example-alpha-M23}
\end{subfigure}\hfill
\begin{subfigure}[t]{0.49\textwidth}
  \centering
  \FPic{YDAncestry-3-1Pgreen1Pbluehatched1}
  \caption{For $i=j=2$ we obtain $M_2$ and $\bm{M}^{22}$ from $\alpha$
  by adding both boxes in row $2$.}
  \label{fig:Example-alpha-M22}
\end{subfigure}
\caption{Two examples constructing $M_i$, $M_j$ and $\bm{M}^{ij}$ from
the Young diagram $\alpha=\ydiagram{3,1}$ for
different values of $i$ and $j$.}
\label{fig:Example-alpha-Mi-Mj-Mij}
\hrule
\end{figure}

We denote the dimensions of the irreps corresponding to $\alpha$, $M_i$, $M_j$
and $\bm{M}^{ij}$ by $d_{\alpha}$, $d_i$, $d_j$, and
$d_{ij}$, respectively,
\begin{equation}
  \label{eq:dimensions-d-Def}
  \text{dim}(\alpha) = d_{\alpha}
  \ , \quad 
  \text{dim}(M_i) = d_i 
  \ , \quad 
  \text{dim}(M_j) = d_j
  \quad \text{and} \quad 
  \text{dim}(\bm{M}^{ij}) = d_{ij}
  \ .
\end{equation} 
These dimensions can be calculated with standard methods from the corresponding
diagrams, see \Cref{app:dimensions}.

Since the $6j$~symbols discussed in this paper have the property that
the two fundamental lines and the irrep indexed by $\alpha$ are fixed,
we may denote the $6j$~symbol in \Cref{eq:6j-Symbol} by $S_{i,j}^{ij}$,
where the bottom two indices $i,j$ correspond to the indices of the
two diagrams $M_i$ and $M_j$, and the upper double-index
corresponds to the double-index of the diagram $\bm{M}^{ij}$,
\begin{equation}
  \label{eq:6j-notation-Def}
  S_{i,j}^{ij}
  =
  \;
    \RPic{6j-Mij-blue-Mi--greenSTAR-Mj-alphaSTAR}
    \ .
\end{equation} 

It should be noted that the irrep $\bm{M}^{ii}$ (both indices equal),
provided it is admissible, is contained in the product
$M_j\otimes\ydiagram{1}$ if and only if $j=i$, as is also illustrated
in \Cref{fig:Example-alpha-M22}. Thus,
\begin{equation}
  \label{eq:6j-upper-indices-equal}
  \text{for fixed $\bm{M}^{ii}$:} \qquad
  S_{a,b}^{ii}
  =
  \delta_{ia}
  \delta_{ib}
  \;
  \FPic{6j-Mii-blue-Ma--greenSTAR-Mb-alphaSTAR}
  \ ,
\end{equation} 
and, similarly, 
\begin{equation}
  \label{eq:6j-upper-indices-fixed}
  \text{for fixed $\bm{M}^{ij}$:} \qquad
  S_{i,b}^{ij}
  =
  \delta_{jb}
  \;
  \FPic{6j-Mij-blue-Mi--greenSTAR-Mb-alphaSTAR}
  \ .
\end{equation} 

Conversely, if we fix the diagrams $M_i$ and $M_j$, the only
irrep labeled by $\bm{M}^{ab}$ that would render the $6j$~symbol
$S_{i,j}^{ab}$ nonzero is precisely that corresponding to $\bm{M}^{ij}$,
\begin{equation}
  \label{eq:6j-lower-indices-fixed}
  \text{for fixed $M_i$ and $M_j$:} \qquad
  S_{i,j}^{ab}
  =
  \delta_{ia}
  \delta_{jb}
  \;
  \FPic{6j-Mab-blue-Mi--greenSTAR-Mj-alphaSTAR}
  \ ,
\end{equation} 
and, if $i=j$,
\begin{equation}
  \label{eq:6j-lower-indices-equal}
  \text{for fixed $M_i$} \qquad
  S_{i,i}^{ab}
  =
  \delta_{ia}
  \;
  \FPic{6j-Mab-blue-Mi--greenSTAR-Mi-alphaSTAR}
  \ .
\end{equation} 

We remark that the $6j$~symbol in \Cref{eq:6j-Symbol} is actually the
most general $6j$ containing two quark-lines not meeting in one
vertex: Even though one might suspect that flipping the direction of
one or several arrows would lead to further $6j$~symbols, this is not
the case as can be deduced from the discussion in
\Cref{sec:properties-2quark-symmetries}. The full set of $6js$ with
two or more quark-lines will, in addition, contain $6j$s where two
quark-lines meet in one vertex. The cases in which both quark-lines
are incoming or in which both quark-lines are outgoing in a vertex are
discussed in \Cref{app:special}. If one of the two quark-lines meeting in
a vertex is incoming and the other is outgoing, we have to
distinguish two cases: Either the two lines form a singlet
(trivial representation), in which case the $6j$ is reduced to a $3j$
(up to normalization), or they are in the adjoint representation,
i.e.\ forming a gluon-line. The latter case will be discussed in a
future publication, together with other $6j$~symbols containing
gluon-lines. The class of $6js$ defined by \eqref{eq:6j-Symbol} is one
that is often encountered in a QCD context,
cf.\ Ref.~\citenum{Sjodahl:2018cca}, Equations.~(2.7) and~(2.8). In future
work we will study the remaining $6j$s needed to decompose color
structure, as well as their applications to QCD.

Before we derive relations between $6j$ symbols of the
form~\eqref{eq:6j-Symbol} in \Cref{sec:relations}, we require some
more properties of these symbols, which will be discussed in the
following \Cref{sec:properties}.

\section{Properties of \texorpdfstring{$6j$}{6j}~symbols}
\label{sec:properties}

The present section discusses several properties of $6j$~symbols that
will be used in this paper. First, we briefly discuss irrep line
orderings in vertices
in~\Cref{sec:line-ordering}. \Cref{sec:properties-graphical-symmetries}
focuses on symmetries of $6j$~symbols, where we, in particular, make
use of the fact that a $6j$~symbol in its graphical representation can
be viewed as a tetrahedron. In \Cref{sec:properties-2quark-symmetries}
we restrict ourselves to the $6j$ symbols of interest in this paper
(i.e.\ those defined in \Cref{sec:6j-alpha-add-two-quark-lines}) and
discuss additional symmetries that arise from having two fundamental
representations on opposite edges of the $6j$.

\subsection{Line ordering in vertices}\label{sec:line-ordering}

In birdtrack calculations we may end up with diagrams that are
complicated to read because of (unnecessary) line crossings. In such
cases it is convenient to introduce barred vertices, indicating that
two lines in a vertex have been swapped, i.e.\ we define
\begin{equation}
  \FPic{GenRep-VertexSTAR-gammaO-betaO-alphaO}
  \; = \; 
  \FPic{GenRep-Vertex-gammaO-betaOSTAR-alphaOSTAR}
  \ .
\end{equation} 
Fortunately, for most vertices appearing in this work, the bars can be
omitted again in the next step as we explain in
\Cref{app:properties-vertices}. The only vertex relevant for this work
for which line swapping leads to a phase change is the vertex that is
used for the antisymmetric projection of two quark-lines, for which we
have
\begin{equation}
  \label{eq:antisymmetrc-two-quark-vertex}
  \FPic{GenRep-Vertex-ASymO-fundISTAR-fundISTAR}
  \; = \;
  - \;
  \FPic{GenRep-Vertex-ASymO-fundI-fundI}
  \ .
\end{equation} 
However, since there are only a few $6j$~symbols satisfying the
conditions laid out in \Cref{sec:6j-alpha-add-two-quark-lines} that
also contain the vertex in \Cref{eq:antisymmetrc-two-quark-vertex}, we
discuss these separately in \Cref{app:special}, and assume for the
remainder of this paper that all vertices appearing in the $6j$~symbols
in question are equal to their barred counterparts,
\begin{equation}
  \label{eq:barred-unbarred-vertices-equal}
  \FPic{GenRep-VertexSTAR-gammaO-betaO-alphaO}
  \; = \; 
  \FPic{GenRep-Vertex-gammaO-betaO-alphaO}
  \hspace{1cm}
  \text{if} \
  \alpha\neq\beta\neq\gamma\neq\alpha
  \ .
\end{equation} 

\subsection{Symmetries of general \texorpdfstring{$6j$}{6j}~symbols}\label{sec:properties-graphical-symmetries}

A Wigner-$6j$~symbol in its graphical \emph{birdtrack} form may be thought of as a
tetrahedron, where our usual notation represents a top-down planar
projection. For example, for a $6j$~symbol connecting general
irreps $\alpha$, $\beta$, $\gamma$, $\delta$, $\rho$ and
$\sigma$, 
  \begin{equation}
    \label{eq:6j-Tetrahedron}
    \RPic{6j-Tetrahedron-delta-rho-sigma--beta-gamma-alpha--CornerLabels}
    \quad \xrightarrow{\text{top-down projection}} \quad 
    \RPic{6j-delta-rho-sigma--beta-gamma-alpha--CornerLabels}
    \ ,
  \end{equation} 
  where we have labeled the vertices of the tetrahedron and the
  corresponding $6j$ as $1,2,3,4$
  for visual clarity. Clearly, which of the four vertices of the tetrahedron we view as the
  top vertex is a completely arbitrary choice and thus cannot
  affect the $6j$~symbol in any way. Therefore, we find that
  \begin{subequations}
    \begin{equation}
      \label{eq:6j-top-corner-symmetry}
      \RPic{6j-delta-rho-sigma--beta-gamma-alpha--CornerLabels}
      \; = \;
      \RPic{6j-rho-gamma-alphaSTAR--sigma-delta-betaSTAR--CornerLabels}
      \; = \;
      \RPic{6j-sigmaSTAR-alpha-betaSTAR--delta-rhoSTAR-gammaSTAR--CornerLabels}
      \; = \;
      \RPic{6j-deltaSTAR-beta-gammaSTAR--rho-sigmaSTAR-alphaSTAR--CornerLabels}
      .
    \end{equation} 
    By that same token, a rotation by $60^{\circ}$ also
    leaves the $6j$~symbol unchanged,
    \begin{equation}
      \label{eq:6j-rotational-symmetry}
      \RPic{6j-delta-rho-sigma--beta-gamma-alpha--CornerLabels}
      \; = \;
      \RPic{6j-sigma-delta-rho--alpha-beta-gamma--CornerLabels}
      \; = \;
      \RPic{6j-rho-sigma-delta--gamma-alpha-beta--CornerLabels}
      \ .
    \end{equation} 
  \end{subequations}

  \subsection{Symmetry properties of \texorpdfstring{$6j$}{6j}~symbols
  with two quark-lines}\label{sec:properties-2quark-symmetries}

  In the present paper, we want to focus on
  $6j$~symbols that were described in \Cref{sec:6j-alpha-add-two-quark-lines},
    \begin{equation}
    \label{eq:6j-Tetrahedron-specialize}
    \RPic{6j-Tetrahedron-delta-rho-sigma--beta-gamma-alpha--CornerLabels}
    \quad \xrightarrow{\text{specialize to}} \quad 
    \RPic{6j-Tetrahedron-Mij-blue-Mi--greenSTAR-Mj-alphaSTAR--CornerLabels}
    \ ,
  \end{equation} 
  which affords us additional symmetries. Firstly, since the two
  quark-lines (the blue and green single lines) are both in the
  \emph{fundamental} representation, we may ``exchange'' them without
  changing the $6j$~symbol,
\begin{equation}
  \label{eq:Sij^ab-reverse-blue-green}
  S_{i,j}^{ij}
  =
  \RPic{6j-Mij-blue-Mi--greenSTAR-Mj-alphaSTAR}
  = 
  \RPic{6j-Mij-green-Mi--blueSTAR-Mj-alphaSTAR}
  \ ;
\end{equation} 
we will, however, continue to draw the two quark-lines in different
colors for visual clarity, as this will make the discussions that
follow (in particular those of \Cref{app:startingpoint}) more legible.

Consider $S_{i,j}^{ab}$ corresponding to the last expression in
\Cref{eq:6j-top-corner-symmetry} and exchange the green and blue
fundamental lines according to \Cref{eq:Sij^ab-reverse-blue-green} to
end up with the following graphical form of $S_{i,j}^{ab}$,
\begin{subequations}
  \label{eq:Sij^ab-Sji^ab-relation-and-Proof}
\begin{equation}
  \label{eq:Sij^ab-Sji^ab-relation-Proof-1}
    S_{i,j}^{ab}
    = \;
    \FPic{6j-Mab-blue-Mi--greenSTAR-Mj-alphaSTAR}
    \xlongequal{\text{\eqref{eq:6j-top-corner-symmetry}}} \;
    \FPic{6j-MabSTAR-greenSTAR-MjSTAR--blue-MiSTAR-alpha}
    \; \xlongequal{\text{\eqref{eq:Sij^ab-reverse-blue-green}}} \; 
    \FPic{6j-MabSTAR-blueSTAR-MjSTAR--green-MiSTAR-alpha}
    \ .
  \end{equation} 
  If we form the complex conjugate of this depiction of the
  $6j$~symbol (which, in the birdtrack formalism, is done by reversing
  all arrows and barring all vertices, cf.
  \Cref{app:properties-vertices}), we will obtain a different
  $6j$~symbol, namely $S_{j,i}^{ab}$ (notice the order of the lower
  indices),
  \begin{equation}
    \label{eq:Sij^ab-Sji^ab-relation-Proof-2}
  (S_{i,j}^{ab})^*
  = \;
  \Vvast( \;
  \FPic{6j-MabSTAR-blueSTAR-MjSTAR--green-MiSTAR-alpha}
  \; \Vvast)^*
  \; = \;
  \FPic{6j-Mab-blue-Mj--greenSTAR-Mi-alphaSTAR--VSTAR1234}
  \; = \;
  \FPic{6j-Mab-blue-Mj--greenSTAR-Mi-alphaSTAR}
  \; =
  S_{j,i}^{ab}
  \ ,
\end{equation} 
\end{subequations}
where we were able to ignore the bars on the vertices in the middle
$6j$~symbol as we assume that all of its vertices obey
\Cref{eq:barred-unbarred-vertices-equal}. (All $6j$ symbols encountered
here, not satisfying this, are the ones involving
\eqref{eq:antisymmetrc-two-quark-vertex}. These are
discussed separately in \Cref{app:special}).

It is easy to convince oneself that all $6j$s naturally can be chosen to
be real by writing out representations in terms of the fundamental
representations and symmetrizers and antisymmetrizers. Contracting
all quark-lines yields a real polynomial in $N$, such that
$S_{i,j}^{ab}\in\mathbb{R}$, and it follows that
$(S_{i,j}^{ab})^*=S_{i,j}^{ab}$, and hence
\begin{equation}
  \label{eq:Sij^ab-Sji^ab-relation}
  S_{i,j}^{ab}
  =
  S_{j,i}^{ab}
  \ .
\end{equation}

In conclusion, for fixed $i,j$ with $i\neq j$, there are four distinct types of $6j$~symbols, namely
\begin{IEEEeqnarray}{rCCCCCCCl}
  & S_{i,i}^{ii} &
  & S_{i,i}^{ij} &
  & S_{j,j}^{ij} &
  & S_{i,j}^{ij} = S_{j,i}^{ij} &
  \nonumber \\
  & \rotatebox[origin=c]{-90}{$=$} & 
  & \rotatebox[origin=c]{-90}{$=$} & 
  & \rotatebox[origin=c]{-90}{$=$} & 
  & \rotatebox[origin=c]{-90}{$=$} & 
  \nonumber \\
  & \RPic{6j-Mii-blue-Mi--greenSTAR-Mi-alphaSTAR} &
  \ , \quad
  & \RPic{6j-Mij-blue-Mi--greenSTAR-Mi-alphaSTAR} &
  \ , \quad
  & \RPic{6j-Mij-blue-Mj--greenSTAR-Mj-alphaSTAR} &
  \quad \text{and} \quad
  & \RPic{6j-Mij-blue-Mi--greenSTAR-Mj-alphaSTAR} &
  \ .
  \label{eq:distinct-6j}
\end{IEEEeqnarray} 
In the following \Cref{sec:closed-form-expressions} we proceed to first derive
relations between these four $6j$~symbols, and then solve this system
of equations to obtain their closed form expressions.

\section{Closed form expressions of~\texorpdfstring{$6j$}{6j}~symbols}
\label{sec:closed-form-expressions}

In the present section, we present several relations between the four
$6j$~symbols given in \Cref{eq:distinct-6j}. In~\Cref{sec:solving}, we
use these relations to find closed form expressions of the
$6j$~symbols. These expressions are summarized
in~\Cref{thm:6j-closed-form-expression}, which is the main result of
this paper.

\subsection{Relations between \texorpdfstring{$6j$}{6j}~symbols}
\label{sec:relations}

Through the repeated use of the completeness
relation \Cref{eq:completeness-relation} and the vertex
correction \Cref{eq:vertex-correction}, we find the following
relations between the four distinct $6j$~symbols given in
\Cref{eq:distinct-6j} (the derivations can be found in
\Cref{app:startingpoint}):
\begin{subequations}
  \label{eq:relations-6j}
  \begin{enumerate}
  \item For a given representation $\bm{M}^{ij}$, we obtain
    \begin{equation}
      \label{eq:relations-6j-1}
      1
      =
      (d_i)^2 (S_{i,i}^{ij})^2
      +
      d_i d_j (S_{i,j}^{ij})^2
       \ .
     \end{equation} 
     Furthermore, 
    \begin{equation}
      \label{eq:relations-6j-2}
      0
      =
      d_i S_{i,i}^{ij} S_{i,j}^{ij}
      +
      d_j S_{i,j}^{ij} S_{j,j}^{ij}
      \ .
    \end{equation} 
    
   \item For two given representations $M_i$ and $M_j$, we obtain
     \begin{equation}
       \label{eq:relations-6j-3}
       \frac{1}{d_{\alpha}}
       =
       \sum_{\bm{M}^{ab}}
       d_{ab} (S_{i,j}^{ab})^2
       \ ,
     \end{equation} 
     where $d_{ab}$ is the dimension of the representation
     $\bm{M}^{ab}$.

   \item For a given representation $M_i$, we have
     \begin{equation}
       \label{eq:relations-6j-linear}
       1
       =
       \sum_b
       d_{ib}
       S^{ib}_{i,i}
       \ .
     \end{equation} 
  \end{enumerate}
\end{subequations}
Notice that, in these relations, all $3j$~symbols were set to $1$; the
derivations in~\Cref{app:startingpoint} keep all $3j$s explicit until
the very last step.

\subsection{Solving for closed form expressions}
\label{sec:solving}

In the present section, we derive closed form expressions for the
$6j$~symbols using the relations presented in the
previous~\Cref{sec:relations}. For the purpose of this section, we
assume that all the $6j$s appearing in these relations are admissible
(that is boxes can be added in rows $i$ and $j$ for the particular
diagram $\alpha$ from which we start).

Let us start with \Cref{eq:relations-6j-1} :
Notice that, if we choose the representation $\bm{M}^{ij}$ such that $i=j$, $\bm{M}^{ij}\to\bm{M}^{ii}$, the second term
vanishes in accordance with \Cref{eq:6j-upper-indices-equal},
such that \Cref{eq:relations-6j-1} reduces to
\begin{IEEEeqnarray}{rCl}
  \text{\Cref{eq:relations-6j-1}}
  \quad
  \xrightarrow{\bm{M}^{ij}\to\bm{M}^{ii}}
  \quad
  1 &=&
  (d_i)^2 (S_{i,i}^{ii})^2
  \nonumber \\
  \Longleftrightarrow \quad
  S_{i,i}^{ii}
  &=&
  \pm
  \frac{1}{d_i}
  \ . \label{eq:Sii-ii-result-up-to-sign}
\end{IEEEeqnarray} 
In fact, in~\Cref{sec:Sii-ii-signs} we show that this $6j$~symbol is
always positive, such that
\begin{equation}
  \label{eq:Sii-ii-result}
  S_{i,i}^{ii} = \frac{1}{d_i}
  \ .
\end{equation}

If in \Cref{eq:relations-6j-3} we instead choose $i\neq j$ (i.e.\ we choose $M_i$ and
$M_j$ to be inequivalent), there is
only one possible $\bm{M}^{ab}$ that renders the $6j$ nonzero, namely
$\bm{M}^{ij}$ (this follows from
condition~\eqref{eq:6j-lower-indices-fixed}). Thus, the
sum on the right hand side of \Cref{eq:relations-6j-3} reduces to
one term, allowing us to solve for yet another $6j$~symbol,
\begin{IEEEeqnarray}{rCl}
  \text{\Cref{eq:relations-6j-3}}
  \quad
  \xrightarrow{i\neq j}
  \quad
  \frac{1}{d_{\alpha}}
  &=&
  d_{ij} (S_{i,j}^{ij})^2
  \nonumber \\
  \Longleftrightarrow \quad
  S_{i,j}^{ij}
  &=&
  \pm \frac{1}{\sqrt{d_{\alpha} d_{ij}}}
  \ . \label{eq:Sij-ij-result}
\end{IEEEeqnarray} 
In~\Cref{sec:Sij-ij-signs} we explain how the overall sign of
$S_{i,j}^{ij}$ can be chosen.

We may now plug the result for $S_{i,j}^{ij}$ back into the full form of \Cref{eq:relations-6j-1} to
also obtain a closed form expression for $S_{i,i}^{ij}$,
\begin{IEEEeqnarray}{rCl}
  1
  =
  (d_i)^2 (S_{i,i}^{ij})^2
  +
  d_i d_j (S_{i,j}^{ij})^2
  \quad
  \xrightarrow{(S_{i,j}^{ij})^2=\frac{1}{d_{\alpha} d_{ij}}}
  \quad
  1
  &=&
  (d_i)^2 (S_{i,i}^{ij})^2
  +
  \frac{d_i d_j}{d_{\alpha} d_{ij}}
  \nonumber \\
  \Longleftrightarrow \quad
  S_{i,i}^{ij}
  &=&
  \pm
  \frac{1}{d_i}
  \sqrt{1-\frac{d_i d_j}{d_{\alpha} d_{ij}}}
  \ . \label{eq:Sii-ij-result}
\end{IEEEeqnarray} 
Lastly, since by \Cref{eq:Sij-ij-result} $S_{i,j}^{ij}\neq0$, we
may now derive the last $6j$~symbol, namely $S_{j,j}^{ij}$, using relation~\eqref{eq:relations-6j-2}:
\begin{IEEEeqnarray}{rCl}
  0
  =
  d_i S_{i,i}^{ij} S_{i,j}^{ij}
  +
  d_j S_{j,j}^{ij} S_{i,j}^{ij}
  & = &
  (d_i S_{i,i}^{ij} + d_j S_{j,j}^{ij})
  S_{i,j}^{ij}
  \nonumber \\
  \xRightarrow{S_{i,j}^{ij}\neq0}
  \quad
  d_j S_{j,j}^{ij}
  & = &
  - d_i S_{i,i}^{ij}
  =
  \mp
  \sqrt{1-\frac{d_i d_j}{d_{\alpha} d_{ij}}}
  \ . \label{eq:Sjj-ij-result}
\end{IEEEeqnarray} 
In~\Cref{sec:Sii-ij-signs} we show how the signs of $S_{i,i}^{ij}$
(and hence $S_{j,j}^{ij}$) can be uniquely determined from
\Cref{eq:relations-6j-linear} for $N\leq3$.

In summary:
\begin{theorem}[Closed form expressions for the distinct $6j$~symbols]\label{thm:6j-closed-form-expression}
  For each of the distinct $6j$~symbols identified
  in~\Cref{eq:distinct-6j}, we obtain the following closed form
  expression:
  \begin{equation}
    \label{eq:6j-closed-form-expression}
    S_{i,i}^{ii}
    =
    \frac{1}{d_i}
    \quad , \qquad
    d_i S_{i,i}^{ij}
    =
    \pm
    \sqrt{1-\frac{d_i d_j}{d_{\alpha} d_{ij}}}
    =
    - d_jS_{j,j}^{ij}
    \quad , \qquad
    S_{i,j}^{ij}
    =
    \pm
    \frac{1}{\sqrt{d_{\alpha} d_{ij}}}
    \ ,
  \end{equation} 
  where the sign of the $S_{i,j}^{ij}$ depends on the definition of
  the vertices as discussed in~\Cref{sec:Sij-ij-signs}.  The signs of
  $S_{i,i}^{ij}$ (equivalently $S_{j,j}^{ij}$) can be uniquely
  determined from \Cref{eq:relations-6j-linear} for $N\leq3$,
  cf.~\Cref{sec:Sii-ij-signs}.
\end{theorem}

We remark that, with the above, the problem of calculating Wigner-$6j$
symbols with quark-lines on opposing edges has been reduced to finding
the dimensions of the representations. Once these are known,
given that a starting representation $\alpha$ results in a maximal
number of new $6j$~symbols, the scaling of finding all
relevant $6j$s (of the given form) with $n$ boxes is therefore given
by the number of possible representations $\alpha$ with $n-2$ boxes,
which scales as $n$ for $N=3$. This implies that finding all
$6j$s with \textit{up to} $n$ boxes scales only as $n^2$ for $N=3$.

\section{Conclusions and outlook}
\label{sec:conclusions}

 In this paper we have taken the first steps towards deriving
  Wigner $6j$~coefficients in terms of $\SUN$ group invariants only
  by writing down closed form expressions of a set of $6j$s involving at least
  two fundamental representations. 

  We are presently supplementing this with a limited set of $6j$s
  involving the adjoint representation, which will be enough to allow
  for a complete color decomposition of amplitudes in QCD
  \cite{Sjodahl:2018cca}.  While the $6j$s with two quark-lines are
  expressed in terms of dimensions only, the gluon $6j$s require the
  quark $6j$s.  Beyond this one may anticipate that yet more general
  $6j$s would be similarly expressible.
  
  Once complemented with the gluon $6j$s, we expect these sets
  of $6j$s to have significant phenomenological relevance
  by opening up, for the first time, the possibility to
  work with orthogonal physical states also for processes
  involving many partons.

\section*{Acknowledgments}
MS acknowledges support by the Swedish Research Council (contract
number 2016-05996, as well as the European Union’s Horizon 2020
research and innovation programme (grant agreement No 668679).  MS and
SP have in part also been supported by the European Union’s Horizon
2020 research and innovation programme as part of the Marie
Sklodowska-Curie Innovative Training Network MCnetITN3 (grant
agreement no. 722104).  JAZ is thankful to the Alexander von
Humboldt Foundation for support via the Fellowship for
Postdoctoral Researchers, as well as to the Erwin Schr{\"o}dinger
Institute for their support via the Junior Research Fellowship.  This
stay at ESI was essential for the completion of this work. We are
also grateful to the Erwin Schr\"odinger Institute Vienna for
hospitality and support while significant parts of this work have been
started within the Research in Teams programme “Amplitude Level
Evolution II: Cracking down on color bases” (RIT0521).

\appendix

\section{Relating different $6j$ symbols}
\label{app:startingpoint}

In this appendix, we derive
\Cref{eq:relations-6j-1,eq:relations-6j-2,eq:relations-6j-3,eq:relations-6j-linear},
which were used to obtain the closed form expressions of the
$6j$~symbols given in \Cref{thm:6j-closed-form-expression}. These
derivations make extensive use of the birdtrack formalism introduced
in \Cref{sec:computational}.

\subsection{Proof of
  \texorpdfstring{\Cref{eq:relations-6j-1,eq:relations-6j-2}}{rels.~a
    and~b}}

Let $\alpha$ be a particular Young diagram, and let $M_i$ and
$\bm{M}^{ij}$ be obtained from $\alpha$ in accordance with
\Cref{sec:6j-alpha-add-two-quark-lines}. Then, we may
consider the following birdtrack diagram,
\begin{equation}
  \label{eq:relations-6j-1and2-LHS}
  \RPic{GenRep-green-alphaVblue--MiVgreen--Mij}
  \ .
\end{equation} 
We may insert a completeness relation (cf.
\Cref{eq:completeness-relation}) between $\alpha$ and the green (top)
quark-line to obtain,
\begin{equation}
  \label{eq:relations-6j-1and2-completeness-relation}
  \RPic{GenRep-green-alphaVblue--MiVgreen--Mij}
  \ = \
  \sum_b \frac{d_b}{\RPic{3j-alphaSTAR-green-Mb}}
  \hspace{3mm}
  \RPic{GenRep-alphaVgreen-Mb-ProjOps--alphaVblue}\hspace{-0.25mm}
  \RPic{GenRep-V--MiVgreen--Mij}
  \ ;
\end{equation} 
from this it is clear that the diagrams $M_b$ are also obtained
from $\alpha$ by adding a single box (corresponding to the top quark-line).

The vertex correction on the right-hand side
of \Cref{eq:relations-6j-1and2-completeness-relation} gives rise to a
$6j$~symbol (cf.\ \Cref{eq:vertex-correction}),
\begin{equation}
  \label{eq:relations-6j-1and2-vertex-correction}
  \RPic{GenRep-Mb--alphaVgreen--alphaVblue--alphaVgreen--Mij}
  \ = \
  \frac{1}{\RPic{3j-blueSTAR-Mb-Mij}} \hspace{3mm}
  \RPic{6j-Mij-blue-Mi--greenSTAR-Mb-alphaSTAR}
  \hspace{3mm}
  \RPic{GenRep-MbVblue--Mij}
  \ ,
\end{equation} 
and allows us to rewrite
\Cref{eq:relations-6j-1and2-completeness-relation} as
\begin{equation}
  \label{eq:relations-6j-1and2-reduced}
  \RPic{GenRep-green-alphaVblue--MiVgreen--Mij}
  \ = \
  \sum_b \frac{d_b}{\RPic{3j-alphaSTAR-green-Mb}
  \ \RPic{3j-blueSTAR-Mb-Mij}}
  \hspace{3mm}
  \RPic{6j-Mij-blue-Mi--greenSTAR-Mb-alphaSTAR}
  \hspace{3mm}
  \RPic{GenRep-alphaVgreen-blue--MbVblue--Mij}
  \ .
\end{equation}

\paragraph{Proof of~\Cref{eq:relations-6j-1}:}\label{sec:relations-6j-1-Derivation}

Consider the Hermitian conjugate of the expression in
\Cref{eq:relations-6j-1and2-LHS} (formed by flipping the birdtrack
about the vertical axis and reversing all arrows, cf.\ Ref.~\citenum{Cvitanovic:2008zz}),
\begin{equation}
  \label{eq:relations-6j-1and2-LHS-HC}
  \left( \;
    \RPic{GenRep-green-alphaVblue--MiVgreen--Mij} \;
  \right)^{\dagger}
  = \;
  \RPic{GenRep-green-alphaVblue--MiVgreen--Mij--HC}
  \ ,
\end{equation} 
and multiply it from the right onto
\Cref{eq:relations-6j-1and2-reduced},
\begin{multline}
  \label{eq:relations-6j-1-reduced}
  \RPic{GenRep-green-alphaVblue--MiVgreen--Mij}\hspace{-0.1mm}
  \RPic{GenRep-green-alphaVblue--MiVgreen--HC}
  \ = \\
  = \
  \sum_b \frac{d_b}{\RPic{3j-alphaSTAR-green-Mb}
  \ \RPic{3j-blueSTAR-Mb-Mij}}
  \hspace{3mm}
  \RPic{6j-Mij-blue-Mi--greenSTAR-Mb-alphaSTAR}
  \hspace{3mm}
  \RPic{GenRep-alphaVgreen-blue--MbVblue--Mij}
  \RPic{GenRep-green-alphaVblue--MiVgreen--HC}
  \ .
\end{multline} 
Let us now take the trace of this equation: the left-hand side yields a product
of $3j$~symbols,
\begin{equation}
  \label{eq:relations-6j-1-LHS-Trace}
  \RPic{GenRep-green-alpha-blue--TraceLeft}
  \RPic{GenRep-green-alphaVblue--MiVgreen--Mij--Trace}\hspace{-0.1mm}
  \RPic{GenRep-green-alphaVblue--MiVgreen--HC-Trace}
  \RPic{GenRep-green-alpha-blue--TraceRight}
  \ = \
  \frac{\RPic{3j-MiSTAR-green-Mij}}{d_i}
  \hspace{3mm}
  \RPic{3j-blueSTAR-alpha-Mi}
  \ ,
\end{equation} 
while the trace of the birdtrack on the right-hand side gives us yet another $6j$~symbol,
\begin{equation}
  \label{eq:relations-6j-1-RHS-Trace}
  \RPic{GenRep-green-alpha-blue--TraceLeft}
  \RPic{GenRep-alphaVgreen-blue--MbVblue--Mij--Trace}\hspace{-0.1mm}
  \RPic{GenRep-green-alphaVblue--MiVgreen--HC-Trace}
  \RPic{GenRep-green-alpha-blue--TraceRight}
  \ =
  \RPic{6j-Mij-green-Mb--blueSTAR-Mi-alphaSTAR--VSTAR1234}
  \ = \
  \RPic{6j-Mij-green-Mb--blueSTAR-Mi-alphaSTAR}
   ;
\end{equation} 
we were able to ignore the bars on all vertices as we assume that all
representations meeting in any particular vertex are distinct,
cf.\ \Cref{eq:barred-unbarred-vertices-equal} (special cases not
obeying this property are discussed separately in \Cref{app:special}).

Putting all of these pieces together, the trace
of \Cref{eq:relations-6j-1-reduced} amounts to the following
expression,
\begin{equation}
  \label{eq:relations-6j-1-symbols}
  \frac{\RPic{3j-MiSTAR-green-Mij}}{d_i}
  \hspace{3mm}
  \RPic{3j-blueSTAR-alpha-Mi}
  \ = \
  \sum_b \frac{d_b}{\RPic{3j-alphaSTAR-green-Mb}
  \ \RPic{3j-blueSTAR-Mb-Mij}}
  \hspace{3mm}
  \underbrace{
    \RPic{6j-Mij-blue-Mi--greenSTAR-Mb-alphaSTAR}
    \hspace{3mm}
    \RPic{6j-Mij-blue-Mb--greenSTAR-Mi-alphaSTAR}
  }_{S^{ij}_{i,b}S^{ij}_{b,i}=(S^{ij}_{i,b})^2}
  \ .
\end{equation} 
Since we are allowed to set all the $3j$~symbols simultaneously to
$1$, this reduces to
\begin{equation}
  \label{eq:relations-6j-1-sum}
  1
  =
  d_i
  \sum_b
  d_b
  (S^{ij}_{ib})^2
  \ .
\end{equation} 
Now, since $\bm{M}^{ij}$ is fixed, the only way for the $6j$~symbol
$S^{ij}_{i,b}$ to be nonzero is if $b\in\lbrace{i,j}\rbrace$
(cf.\ \Cref{eq:6j-upper-indices-fixed}). Hence, the sum on the
right-hand side in \Cref{eq:relations-6j-1-sum} only has two terms,
leaving us with the desired relation~\eqref{eq:relations-6j-1},
\begin{equation}
  \label{eq:relations-6j-1-repeat}
  1
  =
  (d_i)^2
  (S^{ij}_{i,i})^2
  +
  d_i
  d_j
  (S^{ij}_{i,j})^2
  \ .
\end{equation}

\paragraph{Proof of~\Cref{eq:relations-6j-2}:}\label{sec:relations-6j-2-Derivation}

In an analogous way in which we built up the diagram in
\Cref{eq:relations-6j-1and2-LHS}, let us now consider the diagram
\begin{equation}
  \label{eq:relations-6j-2-diagram-multiply}
  \RPic{GenRep-green-alphaVblue--MjVgreen--Mij}
  \ ,
\end{equation} 
such that $i\neq j$ --- in other words $M_i$ and $M_j$ label 
\emph{inequivalent} irreps. Similarly to what we did in
\Cref{eq:relations-6j-1-reduced}, let us take the Hermitian conjugate
of~\eqref{eq:relations-6j-2-diagram-multiply} and multiply it onto
\Cref{eq:relations-6j-1and2-reduced} from the right-hand side,
\begin{multline}
  \label{eq:relations-6j-2-reduced}
  \RPic{GenRep-green-alphaVblue--MiVgreen--Mij}\hspace{-0.1mm}
  \RPic{GenRep-green-alphaVblue--MjVgreen--HC}
  \ = \\
  = \
  \sum_b \frac{d_b}{\RPic{3j-alphaSTAR-green-Mb}
  \ \RPic{3j-blueSTAR-Mb-Mij}}
  \hspace{3mm}
  \RPic{6j-Mij-blue-Mi--greenSTAR-Mb-alphaSTAR}
  \hspace{3mm}
  \RPic{GenRep-alphaVgreen-blue--MbVblue--Mij}
  \RPic{GenRep-green-alphaVblue--MjVgreen--HC}
  \ .
\end{multline} 
Again, we will take the trace of this equation: Since $i\neq j$ (that
is $M_i$ and $M_j$ label inequivalent irreps), the left-hand side
vanishes,
\begin{equation}
  \label{eq:relations-6j-2-LHS-Trace}
  \RPic{GenRep-green-alpha-blue--TraceLeft}
  \RPic{GenRep-green-alphaVblue--MiVgreen--Mij--Trace}\hspace{-0.1mm}
  \RPic{GenRep-green-alphaVblue--MjVgreen--HC-Trace}
  \RPic{GenRep-green-alpha-blue--TraceRight}
  \ = \
  0
  \ .
\end{equation} 
The trace of the birdtrack on the right-hand side once again gives us a $6j$~symbol,
\begin{equation}
  \label{eq:relations-6j-2-RHS-Trace}
  \RPic{GenRep-green-alpha-blue--TraceLeft}
  \RPic{GenRep-alphaVgreen-blue--MbVblue--Mij--Trace}\hspace{-0.1mm}
  \RPic{GenRep-green-alphaVblue--MjVgreen--HC-Trace}
  \RPic{GenRep-green-alpha-blue--TraceRight}
  \ = 
  \RPic{6j-Mij-green-Mb--blueSTAR-Mj-alphaSTAR--VSTAR1234}
  \xlongequal[\eqref{eq:barred-unbarred-vertices-equal}]
  {\eqref{eq:Sij^ab-reverse-blue-green}}
  \RPic{6j-Mij-blue-Mb--greenSTAR-Mj-alphaSTAR}
  \ .
\end{equation} 
Putting all of the pieces together, we obtain the following relation,
\begin{equation}
  \label{eq:relations-6j-2-symbols}
  0
  \ = \
  \sum_b \frac{d_b}{\RPic{3j-alphaSTAR-green-Mb}
  \ \RPic{3j-blueSTAR-Mb-Mij}}
  \hspace{3mm}
  \underbrace{
    \RPic{6j-Mij-blue-Mi--greenSTAR-Mb-alphaSTAR}
    \hspace{3mm}
    \RPic{6j-Mij-blue-Mb--greenSTAR-Mj-alphaSTAR}
  }_{=S^{ij}_{i,b} S^{ij}_{b,j}}
  \ .
\end{equation} 
We will again set the $3j$~symbols to $1$ such that this equation
becomes
\begin{equation}
  \label{eq:relations-6j-2-sum}
  0
  =
  \sum_b
  d_b
  S^{ij}_{i,b} S^{ij}_{b,j}
  \ .
\end{equation} 
Recalling that, since $\bm{M}^{ij}$ is fixed, the only way for
the $6j$~symbols to be nonzero is if $b\in\lbrace{i,j}\rbrace$ in
accordance with \Cref{eq:6j-upper-indices-fixed}. Therefore, the sum
in \Cref{eq:relations-6j-2-sum} again only has two terms, leaving us
with the desired result,
\begin{equation}
  \label{eq:relations-6j-2-repeat}
  0
  =
  d_i
  S^{ij}_{i,i} S^{ij}_{i,j}
  +
  d_j
  S^{ij}_{i,j} S^{ij}_{j,j}
  \ .
\end{equation}

\subsection{Proof of~\texorpdfstring{\Cref{eq:relations-6j-3}}{rel.c}}\label{sec:relations-6j-3-Derivation}

Let $\alpha$ be a particular Young diagram and let $M_i$ and $M_j$ be
obtained from $\alpha$ by adding a box to row $i$ and $j$,
respectively (in accordance with
\Cref{sec:6j-alpha-add-two-quark-lines}). Then, we may consider the
following birdtrack diagram,
  \begin{equation}
  \label{eq:relations-6j-3-LHS}
  \RPic{GenRep-green-alphaVblue--Mi}
  \RPic{GenRep-green-blueValpha--alphaVgreen--Mj--greenValpha}
  \ .
\end{equation} 
We may now insert a completeness relation between $M_i$ and the green (top)
quark-line to obtain
\begin{equation}
  \label{eq:relations-6j-3-completeness-Mi-green}
  \RPic{GenRep-green-alphaVblue--Mi}
  \RPic{GenRep-green-blueValpha--alphaVgreen--Mj--greenValpha}
  \ = \
  \sum_{\bm{M}^{ab}}
  \frac{d_{ab}}{\RPic{3j-MiSTAR-green-Mab}}
  \hspace{3mm}
  \RPic{GenRep-green-alphaVblue}
  \hspace{-0.2mm}
  \RPic{GenRep-V--MiVgreen--Mab--ProjOps}
  \RPic{GenRep-green-blueValpha--alphaVgreen--Mj--greenValpha}
  \ .
\end{equation} 
On the right hand side of this equation we obtain a $6j$~symbol from
the vertex correction,
\begin{IEEEeqnarray}{rCl}
  \label{eq:relations-6j-3-vertex-correction}
  \RPic{GenRep-Mab--MiVgreen--alphaVblue--alphaVgreen--Mj-blue}
  & = &
  \frac{1}{\RPic{3j-blueSTAR-Mj-Mab}} \hspace{3mm}
  \RPic{6j-Mab-green-Mj--blueSTAR-Mi-alphaSTAR--VSTAR1234}
  \RPic{GenRep-Mab-MjVblue}
  \nonumber
  \\
  & \ \xlongequal[\eqref{eq:barred-unbarred-vertices-equal}]
  {\eqref{eq:Sij^ab-reverse-blue-green}} \ &
  \frac{1}{\RPic{3j-blueSTAR-Mj-Mab}} \hspace{3mm}
  \RPic{6j-Mab-blue-Mj--greenSTAR-Mi-alphaSTAR}
  \RPic{GenRep-Mab-MjVblue}
  \ ,
\end{IEEEeqnarray} 
such that \Cref{eq:relations-6j-3-completeness-Mi-green} reduces to
\begin{multline}
  \label{eq:relations-6j-3-RHS-reduced}
  \RPic{GenRep-green-alphaVblue--Mi}
  \RPic{GenRep-green-blueValpha--alphaVgreen--Mj--greenValpha}
  \ = \\
  = \
  \sum_{\bm{M}^{ab}}
  \frac{d_{ab}}{\RPic{3j-MiSTAR-green-Mab}
  \ \RPic{3j-blueSTAR-Mj-Mab}}
  \hspace{3mm} 
  \RPic{6j-Mab-blue-Mj--greenSTAR-Mi-alphaSTAR}
  \hspace{3mm}
  \RPic{GenRep-green-alphaVblue--MiVgreen--Mab}
  \RPic{GenRep-MjVblue--greenValpha}
  \ .
\end{multline} 
Let us now take the trace of
\Cref{eq:relations-6j-3-RHS-reduced}: When tracing the birdtrack
diagram on the left-hand side, we simply get a product of $3j$~symbols with a
dimension factor,
\begin{equation}
  \label{eq:relations-6j-3-LHS-Trace}
  \RPic{GenRep-green-alpha-blue--TraceLeft}
  \RPic{GenRep-green-alphaVblue--Mi--Trace}
  \RPic{GenRep-green-blueValpha--alphaVgreen--Mj--greenValpha--Trace}
  \RPic{GenRep-green-alpha-blue--TraceRight}
  \ = \
  \frac{\RPic{3j-alphaSTAR-green-Mj}}{d_{\alpha}}
  \hspace{3mm}
  \RPic{3j-blueSTAR-alpha-Mi}
  \ .
\end{equation} 
The trace of the birdtrack on the right-hand side
of \Cref{eq:relations-6j-3-RHS-reduced} yields another $6j$~symbol,
\begin{equation}
  \label{eq:relations-6j-3-RHS-Trace}
  \RPic{GenRep-green-alpha-blue--TraceLeft}
  \RPic{GenRep-green-alphaVblue--MiVgreen--Mab--Trace}
  \RPic{GenRep-MjVblue--greenValpha--Trace}
  \RPic{GenRep-green-alpha-blue--TraceRight}
  \ = \
  \RPic{6j-Mab-blue-Mi--greenSTAR-Mj-alphaSTAR}
  \ .
\end{equation} 
Substituting expressions~\eqref{eq:relations-6j-3-LHS-Trace}
and~\eqref{eq:relations-6j-3-RHS-Trace} back into the traced
\Cref{eq:relations-6j-3-RHS-reduced} yields
\begin{equation}
  \label{eq:relations-6j-3-Trace}
  \frac{\RPic{3j-alphaSTAR-green-Mj}}{d_{\alpha}}
  \hspace{3mm}
  \RPic{3j-blueSTAR-alpha-Mi}
  \ = \
  \sum_{\bm{M}^{ab}}
  \frac{d_{ab}}{\RPic{3j-MiSTAR-green-Mab}
  \ \RPic{3j-blueSTAR-Mj-Mab}}
  \hspace{3mm}
  \underbrace{
    \RPic{6j-Mab-blue-Mj--greenSTAR-Mi-alphaSTAR}
    \hspace{3mm}
    \RPic{6j-Mab-blue-Mi--greenSTAR-Mj-alphaSTAR}
  }_{S_{j,i}^{ab}S_{i,j}^{ab}=(S_{i,j}^{ab})^2}
  ,
\end{equation} 
where we used \Cref{eq:Sij^ab-Sji^ab-relation} to write
$S_{j,i}^{ab}S_{i,j}^{ab}=(S_{i,j}^{ab})^2$ ($6j$s for which this
relation does not hold are discussed separately in \Cref{app:special}).
We once again use the fact that we may set all the $3j$~symbols
simultaneously to $1$ to finally obtain the desired
\Cref{eq:relations-6j-3},
\begin{equation}
  \label{eq:relations-6j-3-Repeat}
  \frac{1}{d_{\alpha}}
  =
  \sum_{\bm{M}^{ab}} d_{ab} (S_{i,j}^{ab})^2
  \ ,
\end{equation} 
again with the only exception of $6j$-symbols involving the antisymmetric
vertex in \cref{eq:antisymmetrc-two-quark-vertex}.

\subsection{Proof of the linear
  relation~\texorpdfstring{\Cref{eq:relations-6j-linear}}{d}}\label{sec:relations-6j-linear-Derivation}

Let $M_i$ be a particular Young diagram obtained from $\alpha$ by
adding a single box, and consider the following birdtrack diagram

\begin{equation}
  \label{eq:relations-6j-linear-LHS}
  \RPic{GenRep-green-alphaVblue--Mi--blueValpha}
  \ .
\end{equation} 
Let us now insert a completeness relation (cf.
\Cref{eq:completeness-relation}) between $M_i$ and the
green (top) quark-line,
\begin{equation}
  \label{eq:relations-6j-linear-completeness-relation}
  \RPic{GenRep-green-alphaVblue--Mi--blueValpha}
  \ =
  \sum_{\bm{M}^{ab}}
  \frac{d_{ab}}{\RPic{3j-MiSTAR-green-Mab}}
  \hspace{3mm}
  \RPic{GenRep-green-alphaVblue}
  \hspace{-0.2mm}
  \RPic{GenRep-V--MiVgreen--Mab--ProjOps}
  \RPic{GenRep-green-alphaVblue--HC}
  \ .
\end{equation} 
If we were to take a trace of this equation, we would obtain a bunch
of $3j$~symbols, which is not particularly interesting. However,
instead of merely taking a trace, let us first ``swap'' the two
quark-lines (i.e.\ we multiply
\Cref{eq:relations-6j-linear-completeness-relation} with a
transposition between the two quark-lines from the right). This is a
perfectly legal thing to do as the two quark-lines are both in the
fundamental representation by definition,
\begin{equation}
  \label{eq:relations-6j-linear-swap}
  \RPic{GenRep-green-alphaVblue--Mi--blueValpha}
  \RPic{GenRep-greenSblue-alpha}
  \ =
  \sum_{\bm{M}^{ab}}
  \frac{d_{ab}}{\RPic{3j-MiSTAR-green-Mab}}
  \hspace{3mm}
  \RPic{GenRep-green-alphaVblue}
  \hspace{-0.2mm}
  \RPic{GenRep-V--MiVgreen--Mab--ProjOps}
  \RPic{GenRep-green-alphaVblue--HC}
  \RPic{GenRep-greenSblue-alpha}
  \ .
\end{equation} 
If we now take the trace of this equation, the left-hand side will still yield a $3j$~symbol,
\begin{equation}
  \label{eq:relations-6j-linear-LHS-Trace}
  \RPic{GenRep-green-alpha-blue--TraceLeft}
  \RPic{GenRep-green-alphaVblue--Mi--blueValpha--Trace}
  \RPic{GenRep-greenSblue-alpha--Trace}
  \RPic{GenRep-green-alpha-blue--TraceRight}
  \ = \
  \RPic{3j-blueSTAR-alpha-Mi}
  \ ,
\end{equation} 
but the right-hand side yields a $6j$~symbol,
\begin{equation}
  \label{eq:relations-6j-linear-RHS-Trace}
  \RPic{GenRep-green-alpha-blue--TraceLeft}
  \RPic{GenRep-green-alphaVblue--Trace}
  \hspace{-0.2mm}
  \RPic{GenRep-V--MiVgreen--Mab--ProjOps--Trace}
  \hspace{-0.2mm}
  \RPic{GenRep-green-alphaVblue--HC--Trace}
  \RPic{GenRep-greenSblue-alpha--Trace}
  \RPic{GenRep-green-alpha-blue--TraceRight}
  \ = \
  \RPic{6j-Mab-green-Mi--blueSTAR-Mi-alphaSTAR--VSTAR24}
  \ =
  S^{ab}_{i\overline{\imath}}
  \ .
\end{equation} 
The symbol $S^{ab}_{i,\overline{\imath}}$ is similar to the $6j$~symbol $S^{ab}_{i,i}$,
except the bar over one of the indices, $\overline{\imath}$, indicates that the
vertices adjacent to one representation line $M_i$ have been conjugated.

Putting the pieces together, we find that
\begin{equation}
  \label{eq:relations-6j-linear-Trace}
  \RPic{3j-blueSTAR-alpha-Mi}
  \ =
  \sum_{\bm{M}^{ab}}
  \frac{d_{ab}}{\RPic{3j-MiSTAR-green-Mab}}
  \hspace{3mm}
  \underbrace{
    \RPic{6j-Mab-green-Mi--blueSTAR-Mi-alphaSTAR--VSTAR24}
  }_{S^{ab}_{i,\overline{\imath}}}
   \ .
 \end{equation} 
 Once again, we ignore the conjugated vertices on the $6j$~symbol
 $S^{ab}_{i,\overline{\imath}}$ as we assume that
 \Cref{eq:barred-unbarred-vertices-equal} holds, and refer the reader
 to \Cref{app:special} for all $6j$~symbols for which the
 assumption~\eqref{eq:barred-unbarred-vertices-equal} is not
 valid. Thus, we have that
 $S^{ab}_{i,\overline{\imath}}=S^{ab}_{i,i}$, and
 \Cref{eq:relations-6j-linear-Trace} reduces to
\begin{equation}
  \RPic{3j-blueSTAR-alpha-Mi}
  \ =
  \sum_{\bm{M}^{ab}}
  \frac{d_{ab}}{\RPic{3j-MiSTAR-green-Mab}}
  \hspace{3mm}
  \underbrace{
    \RPic{6j-Mab-green-Mi--blueSTAR-Mi-alphaSTAR}
  }_{S^{ab}_{i,i}}
  \ .
 \end{equation} 
 Lastly, setting all $3j$~symbols to $1$, we obtain the desired
 \Cref{eq:relations-6j-linear},
\begin{equation}
  \label{eq:relations-6j-linear-Repeat}
  1
  =
  \sum_b
  d_{ib}
  S^{ib}_{i,i}
  \ ,
\end{equation} 
where we used the fact that at least one of the indices $a,b$ (which
one doesn't matter due to \Cref{eq:Mij=Mji}) must be equal to $i$ for
the $6j$ symbol to be nonzero, cf.
\Cref{eq:6j-lower-indices-equal}.

\section{Fixing the sign ambiguity}
\label{app:signs}

In this section, we will determine the signs for the $6j$~symbols
$S_{i,i}^{ii}$, $S_{i,j}^{ij}$ and
$S_{i,i}^{ij}=-\frac{d_j}{d_i}S_{j,j}^{ij}$ in
\Cref{sec:Sii-ii-signs,sec:Sij-ij-signs,sec:Sii-ij-signs},
respectively. For the first two cases ($S_{i,i}^{ii}$ and
$S_{i,j}^{ij}$), the overall sign can be determined in an
$N$-independent way. For the last case
($S_{i,i}^{ij}=-\frac{d_j}{d_i}S_{j,j}^{ij}$) we are able to determine
the signs uniquely for $N\leq3$, but argue that further work is needed
to reliably determine the signs for $N>3$.  \ytableausetup{boxsize=3mm}

\subsection{Fixing the sign
  of~\texorpdfstring{$S_{i,i}^{ii}$}{$Sii-ii$}}\label{sec:Sii-ii-signs}

We start by determining the sign of $S_{i,i}^{ii}$: Let us now recall
that, by the definition of the $6j$~symbols given in
\Cref{sec:6j-alpha-add-two-quark-lines}, $i$ denotes the row of
$\alpha$ at the end of which the new box was added. In particular,
this means that the two boxes were added to the \emph{same} row for
the symbol $S_{i,i}^{ii}$, implying that the two quark-lines enter symmetrically
in $M^{ii}$. Let us now re-draw the 6j~symbol somewhat:
\begin{equation}
  \label{eq:Sii^ii-fix-signs-1}
  S_{i,i}^{ii}
  =
  \RPic{6j-Mii-blue-Mi--greenSTAR-Mi-alphaSTAR}
  =
  \RPic{6jsquare-Mii-Mi-alpha-Mi--green-blue--VSTAR23}
  \; \xlongequal{\text{\eqref{eq:barred-unbarred-vertices-equal}}} \;
  \RPic{6jsquare-Mii-Mi-alpha-Mi--green-blue}
  \ .
\end{equation} 
Rewriting the crossed fundamental lines in the last birdtrack as a sum
of symmetrizers and antisymmetrizers, $\FPic{2s12SNArr} \; = \;
  \FPic{2ArrLeft}
  \FPic{2Sym12SN}
  \FPic{2ArrRight}
  \; - \;
  \FPic{2ArrLeft}
  \FPic{2ASym12SN}
  \FPic{2ArrRight}$ (see, for example, Ref.~\citenum{Cvitanovic:2008zz}), we obtain
\begin{equation}
  \label{eq:Sii^ii-fix-signs-2}
  S_{i,i}^{ii}
  = \;
  \RPic{6jsquare-Mii-Mi-alpha-Mi--green-blue}
  \; = \;
  \RPic{6jsquare-Mii-Mi-alpha-Mi--green-blue-Sym12}
  \; - \;
  \RPic{6jsquare-Mii-Mi-alpha-Mi--green-blue-ASym12}
  \ .
\end{equation} 
Since the two quark-lines enter symmetrically in $M^{ii}$, the last
term in \Cref{eq:Sii^ii-fix-signs-2} vanishes. More precisely,
rewriting the triple product
$\alpha\otimes\ydiagram{1}\otimes\ydiagram{1}$ as
$\left(\alpha\otimes\ydiagram{2}\right) \oplus
\left(\alpha\otimes\ydiagram{1,1}\right)$ and decomposing the result
into irreps, one can check that $M^{ii}$ appears in
$\alpha\otimes\ydiagram{2}$ but not in
$\alpha\otimes\ydiagram{1,1}$. Thus, writing
$\FPic{2ArrLeft} \FPic{2Sym12SN} \FPic{2ArrRight} \; = \; \frac{1}{2}
\left( \; \FPic{2IdSNArr} \; + \; \FPic{2s12SNArr} \; \right)$,
\begin{equation}
  \label{eq:Sii^ii-fix-signs-3}
  S_{i,i}^{ii}
  = \;
  \RPic{6jsquare-Mii-Mi-alpha-Mi--green-blue}
  \; = 
  \frac{1}{2}
  \left( 
    \RPic{6jsquare-Mii-Mi-alpha-Mi--green-blue-s12}
    \; + \;
    \RPic{6jsquare-Mii-Mi-alpha-Mi--green-blue}
  \right)
  \ .
\end{equation} 
Recognizing the second term of the right-hand side as $\frac{1}{2}S_{i,i}^{ii}$ and
taking it to the left-hand side, we notice that the right-hand side reduces to a product of
$3j$~symbols with a dimension factor,
\begin{equation}
  \label{eq:Sii^ii-fix-signs-4}
  S_{i,i}^{ii}
  = \;
  \RPic{6jsquare-Mii-Mi-alpha-Mi--green-blue}
  \; = \;
  \RPic{6jsquare-Mii-Mi-alpha-Mi--green-blue-s12}
  \; = \;
  \frac{\RPic{3j-MiSTAR-greenblue-Mii}}{d_i}
  \cdot
  \RPic{3j-greenblueSTAR-Mi-alpha}
  \ .
\end{equation} 
Again setting the $3j$~symbols to $1$, we obtain that
\begin{equation}
  \label{eq:1}
  S_{i,i}^{ii}
  =
  \frac{1}{d_i}
  \ ,
\end{equation} 
now with a definite sign. We comment that, in determining the sign of
$S_{i,i}^{ii}$, we actually rederived $S_{i,i}^{ii}$ with a definite
sign. However, since similar methods will not work for 
$S_{i,j}^{ii}$, $S_{i,j}^{jj}$, and $S_{i,j}^{ij}$, we view this as a
consistency check and continue to use Equations~\eqref{eq:relations-6j} to
derive the functional forms of the remaining $6j$~symbols.

\subsection{Fixing the sign
  of~\texorpdfstring{$S_{i,j}^{ij}$}{$Sij-ij$} for $i\ne j$}
\label{sec:Sij-ij-signs}

Consider the graphical notation for $S_{i,j}^{ij}$,
\begin{equation}
  \label{eq:Sij-ij-signs-graphical-notation}
  S_{i,j}^{ij}
  =
  \RPic{6j-Mij-blue-Mi--greenSTAR-Mj-alphaSTAR--CornerLabels}
  \ .
\end{equation} 
Notice that each vertex occurs exactly once in this $6j$~symbol,
\begin{equation}
  \label{eq:Sij-ij-signs-graphical-notation-vertices}
  \RPic{6jVertex-greenSTAR-Mj-alphaSTAR--1}
  \ , \quad
  \RPic{6jVertex-Mij-blueSTAR-Mj--2}
  \ , \quad 
  \RPic{6jVertex-Mi-MijSTAR-greenSTAR--3}
  \ , \quad
  \RPic{6jVertex-blue-MiSTAR-alpha--4}
  \ .
\end{equation} 
The overall sign of the $6j$~symbol $S_{i,j}^{ij}$ is uniquely
determined by how one decided to define the vertices that occur in
this $6j$~symbol. When iteratively computing these $6j$~symbols,
starting with small Young diagrams and then adding more and more
boxes, we may encounter vertices that were already used for earlier
$6j$s. If all four vertices have been encountered earlier then the
sign of the $6j$~symbol may already be fixed. Otherwise we can pick
the sign of the $6j$ to be, say, positive, thus imposing a constraint
on signs of the vertices. Often a pair of vertices (either vertices 1
and 4 or vertices 2 and 3) is encountered for the first time in the
$6j$~symbol to compute. Then, picking the sign of the $6j$ puts a
constraint on the product of the signs of the newly encountered vertex
pair.
Since fully contracted color structures consist of dimensions,
$3j$~symbols and $6j$~symbols \emph{only}, the information used and
obtained in an iterative computation of $6j$~symbols is sufficient to
perform calculations in color space.

\subsection{Fixing the sign
  of~\texorpdfstring{$S_{i,i}^{ij}$}{$Sii-ij$} for $i \ne j$
  (equivalently~\texorpdfstring{$S_{j,j}^{ij}$}{$Sjj-ij$})}\label{sec:Sii-ij-signs}

Lastly, we turn to the $6j$~symbols $S_{i,i}^{ij}$ and, equivalently,
$S_{j,j}^{ij}$, whose functional form is given in
\Cref{thm:6j-closed-form-expression}. We notice that, for these
$6j$~symbols, each vertex occurs together with its complex conjugated
version. Thus, merely the vertex definitions do not determine the overall
sign but a different method has to be chosen. 

In the present section we will discuss how the linear
relation~\eqref{eq:relations-6j-linear} can be used to fix the signs
of $S_{i,i}^{ij}$ and $S_{j,j}^{ij}$ for $N\leq3$, and comment on why additional
work is needed to reliably fix the signs beyond $N=3$. However, since our
focus lies on physics applications (which will be discussed in a
future paper), it is sufficient to fix the signs for $N=3$, where
$N$ is interpreted as the number of colors $N_c$.

We require two preliminary results:

\begin{lemma}[Determining the relative signs in a sum]\label{lem:sum-nonzero-subsum}
  Consider the set of \emph{known}, positive, real numbers
  $\lbrace{A_i}\rbrace_{i=1}^{k}$, and suppose that
  \begin{equation}
    \label{eq:sum-nonzero-subsum-Cond1}
    \sum_{i=1}^k \chi_i
      A_i
    =
    C
    \ ,
  \end{equation} 
  where $C\in\mathbb{R}\setminus\lbrace0\rbrace$ is also known, and
  the $\chi_i\in\lbrace{-1,1}\rbrace$ are to be determined. Then, if
  all subsets of $\lbrace{A_i}\rbrace_{i=1}^{k}$ satisfy
  \begin{equation}
    \label{eq:sum-nonzero-subsum-Cond2}
    \sum_{l=1}^{m\leq k} \chi_{j_l}
      A_{j_l}
    \neq
    0
    \ ,
  \end{equation} 
  all the $\chi_i$ can be determined uniquely.
\end{lemma}

\begin{proof}[\Cref{lem:sum-nonzero-subsum}]
  We present a proof by contradiction: Let the
  $\left\lbrace{A_i}\right\rbrace_{i=1}^k$ be such that
  conditions~\eqref{eq:sum-nonzero-subsum-Cond1}
  and~\eqref{eq:sum-nonzero-subsum-Cond2} laid out in the lemma are
  satisfied. Suppose now that \Cref{eq:sum-nonzero-subsum-Cond1} does
  not uniquely determine all the $\chi_i\in\lbrace-1,1\rbrace$, that
  is, there exist some
  $\lbrace{j_1,\ldots,j_m}\rbrace\subset\lbrace{1,\ldots,k}\rbrace$
  such that
  \begin{subequations}
    \begin{equation}
      \label{eq:sum-nonzero-subsum-Proof1a}
      \phantom{\text{and} \hspace{2cm}}
      \sum_{i\in\lbrace{1,\ldots,k}\rbrace\setminus\lbrace{j_1,\ldots,j_m}\rbrace}
      \chi_i
      A_i
      +
      \sum_{i\in\lbrace{j_1,\ldots,j_m}\rbrace}
      \chi_i
      A_i
      =
      C
      \phantom{\ .}
    \end{equation} 
    \begin{equation}
      \label{eq:sum-nonzero-subsum-Proof1b}
      \text{and} \hspace{2cm}
      \sum_{i\in\lbrace{1,\ldots,k}\rbrace\setminus\lbrace{j_1,\ldots,j_m}\rbrace}
      \chi_i
      A_i
      -
      \sum_{i\in\lbrace{j_1,\ldots,j_m}\rbrace}
      \chi_i
      A_i
      =
      C
      \ .
    \end{equation} 
  \end{subequations}
  Then, deducting \Cref{eq:sum-nonzero-subsum-Proof1b} from
  \Cref{eq:sum-nonzero-subsum-Proof1a}, we obtain
  \begin{equation}
    \label{eq:sum-nonzero-subsum-Proof2}
    \sum_{i\in\lbrace{j_1,\ldots,j_m}\rbrace}
      \chi_i
      A_i
      =
      0
      \ .
    \end{equation} 
    Thus, we have found a partial sum of the products $\chi_iA_i$ that
    vanishes, which poses a contradiction to~\Cref{eq:sum-nonzero-subsum-Cond2}.
\end{proof}

\begin{lemma}[Adding a box to the first row of a Young diagram]\label{lem:d11-d1-fraction-Nc-restriction}
  Let $\alpha$ be a Young diagram and let $M_1$ and $\bm{M}^{11}$ be
  the diagrams obtained from $\alpha$ by adding one, respectively two,
  box(es) to the first row of $\alpha$. Furthermore, let $d_1$ and
  $d_{11}$ be the dimensions of the irreducible representations
  corresponding to $M_1$ and $\bm{M}^{11}$, respectively. Then,
  \begin{equation}
    \label{eq:d11-d1-fraction-Nc-restriction}
    \frac{d_{11}}{d_1} = 1
    \quad
    \Rightarrow
    \quad
    N \leq 1
    \ ,
  \end{equation} 
  where equality holds if and only if $\alpha$ is the totally
  symmetric diagram consisting of exactly one row.
\end{lemma}

\begin{proof}[\Cref{lem:d11-d1-fraction-Nc-restriction}]
  Let $h_{a,b}$ be the hook lengths (see \Cref{app:dimensions}) of
  $\alpha$, and denote by $\ell$ the length of $\alpha$'s first
  row. We calculate $d_{11}/d_1$ using the factors-over-hooks formula
  \eqref{eq:factors-over-hooks}. For $\bm{M}^{11}$, compared to $M_1$,
  we obtain just one additional factor, namely $N+\ell+1$ for the last
  box in the first row, and only the hook lengths for the boxes in the
  first row differ. In the quotient $d_{11}/d_1$ all other hook
  lengths cancel, i.e.\
  \begin{equation}
    \frac{d_{11}}{d_1} = (N+\ell+1) \,
    \frac{(h_{1,1}+1) (h_{1,2}+1) \cdots (h_{1,\ell}+1)\cdot 1}
         {(h_{1,1}+2) (h_{1,2}+2) \cdots (h_{1,\ell}+2)\cdot 2} \ .
  \end{equation} 
  Hence, 
  \begin{equation}
    d_{11}/d_1=1
    \quad \Leftrightarrow \quad
    N = \frac{(h_{1,1}+2) (h_{1,2}+2) \cdots (h_{1,\ell}+2)}
             {(h_{1,1}+1) (h_{1,2}+1) \cdots (h_{1,\ell}+1)} \cdot 2
    - \ell - 1 \, .
  \end{equation} 
  The quotients $\frac{h_{1,b}+2}{h_{1,b}+1}$ are maximal if $h_{1,b}$
  is minimal, and $h_{1,b} \geq \ell-b+1$, where equality holds if and
  only if there is only one box in column $b$. Therefore,
  \begin{equation}
    N \leq
    \frac{(\ell+2)\cancel{(\ell+1)}\cdots \cancel{3}}
         {\cancel{(\ell+1)} \ \cdots \ \cancel{3}
    \cdot \bcancel{2}} \cdot \bcancel{2} - \ell -1 = 1 \ ,
\end{equation} 
as required.
\end{proof}

Let's manipulate the linear relation~\eqref{eq:relations-6j-linear} a
little bit: First, we single out the known value
$S^{ii}_{i,i}=\frac{1}{d_i}$ from the sum, 
\begin{equation}
  \label{eq:relations-6j-linear-rewrite-1}
  1
  =
  \sum_b
  d_{ib}
  S^{ib}_{i,i}
  =
  \frac{d_{ii}}{d_i}
  +
  \sum_{b\neq i} d_{ib} S^{ib}_{i,i}
  \ .
\end{equation}

We would like to view the signs of the $6j$~symbols as 
variables to be determined, and therefore define $\chi_{ij}$
\begin{equation}
  \label{eq:chiij-Def}
  \chi_{ij} = \chi_{ji} \in \lbrace-1,1\rbrace
  \ ,
\end{equation} 
such that
\begin{equation}
  \label{eq:6j-with-chi}
  S_{i,i}^{ij} =
  \chi_{ij} \frac{1}{d_i}
  \sqrt{1-\frac{d_i d_j}{d_{\alpha} d_{ij}}}
  \ .
\end{equation} 
Notice that $\chi_{ij}$ does \emph{not} encompass the relative sign
between $S_{i,i}^{ij}$ and $S_{j,j}^{ji}$ but denotes the
\emph{absolute} sign of $S_{i,i}^{ij}$. In other words,
\begin{equation}
  \label{eq:chi-absolute-sign}
  S_{i,i}^{ij} =
  \chi_{ij} \frac{1}{d_i}
  \sqrt{1-\frac{d_i d_j}{d_{\alpha} d_{ij}}}
  \xlongequal{\eqref{eq:6j-closed-form-expression}}
  -
  \chi_{ij} \frac{1}{d_j}
  \sqrt{1-\frac{d_i d_j}{d_{\alpha} d_{ij}}}
  =
  - \frac{d_j}{d_i}
  S_{j,j}^{ji}
  \ .
\end{equation} 
Furthermore, to make the notation a bit shorter, let us define
the symbol $A_{ij}$ as
\begin{equation}
  \label{eq:Aij-Def}
  A_{ij} =
  \frac{d_{ij}}{d_i}
  \sqrt{1-\frac{d_i d_j}{d_{\alpha} d_{ij}}}
  \hspace{1cm}
  \Longrightarrow
  \hspace{1cm}
 \chi_{ij} A_{ij} = d_{ij}S_{i,i}^{ij}
  \ .
\end{equation} 
Clearly, the $d_{ij}$ are symmetric in $i$ and $j$, and so are the
$S_{i,i}^{ij}$ in their upper indices by
\Cref{eq:Mij=Mji},
$S_{i,i}^{ij}=S_{i,i}^{ji}$. However, $S_{i,i}^{ij}$ and $S_{j,j}^{ij}$
are related by a negative pre-factor according to
\Cref{eq:6j-closed-form-expression} in
\Cref{thm:6j-closed-form-expression}, such that the $A_{ij}$ are antisymmetric,
\begin{equation}
  \label{eq:Aij-antisymmetric}
  A_{ji} = - A_{ij}
  \ .
\end{equation}

Then, taking $\frac{d_{ii}}{d_i}$ to the other side of the equal sign
and implementing notation~\eqref{eq:Aij-Def}, 
the linear equations in~\eqref{eq:relations-6j-linear-rewrite-1} can
be cast into matrix form,
\begin{equation}
  \label{eq:relations-6j-linear-rewrite-matrix}
  \begin{pmatrix}
    0 & \chi_{12}A_{12} & \chi_{13}A_{13} & \ldots & \chi_{1N}A_{1N} \\
    -\chi_{12}A_{12} & 0 & \chi_{23}A_{23} & \ldots & \chi_{2N}A_{2N} \\
    -\chi_{13}A_{13}
    & -\chi_{23}A_{23} & 0
    & \ldots & \chi_{3N}A_{3N} \\
    \vdots & \vdots & \vdots & \ddots & \vdots \\
    -\chi_{1N}A_{1N}
    & -\chi_{2N}A_{2N}
    & -\chi_{3N}A_{3N} & \ldots & 0
  \end{pmatrix}
  \begin{pmatrix}
    1 \\
    1 \\
    1 \\
    \vdots \\
    1 \\
  \end{pmatrix}
  =
  \begin{pmatrix}
    1-\frac{d_{11}}{d_1} \\
    1-\frac{d_{22}}{d_2} \\
    1-\frac{d_{33}}{d_3} \\
    \vdots \\
    1-\frac{d_{{N}{N}}}{d_{N}}
  \end{pmatrix}
  \ .
\end{equation} 
Recall that the $A_{ij}$ are known (as the dimensions $d_{\alpha}$,
$d_i$, $d_j$ and $d_{ij}$ are known, cf.\ \Cref{eq:Aij-Def}) and that
we seek to determine the $\chi_{ij}\in\lbrace{-1,1}\rbrace$ for each
pair $(i,j)$.  We shall denote the linear equation resulting from row
$r$ of the matrix
equation~\eqref{eq:relations-6j-linear-rewrite-matrix} by $E(r)$, that
is:
\begin{equation}
  \label{eq:Er-Def}
  E(r): \quad
  - \sum_{i=1}^{r-1} \chi_{ir} A_{ir}
  +
  \sum_{j=r+1}^{N} \chi_{rj} A_{rj}
  =
  1 - \frac{d_{rr}}{d_r}
  \ .
\end{equation} 
Notice that, up to this point, we have not specified a particular
value for $N$ but kept the discussion fully general. From now on, let
us fix $N=3$. (We note that the below argument also works for $N<3$.
At the end of this section, we comment on why the strategy presented
here for determining the signs breaks down for $N>3$.)

\paragraph{For $N=3$,} the matrix
equation~\eqref{eq:relations-6j-linear-rewrite-matrix} simplifies as
\begin{equation}
  \label{eq:NC3-relations-6j-linear-matrix}
  \begin{pmatrix}
    0 & \chi_{12}A_{12} & \chi_{13}A_{13} \\
    -\chi_{12}A_{12} & 0 & \chi_{23}A_{23} \\
    -\chi_{13}A_{13} & -\chi_{23}A_{23} & 0
  \end{pmatrix}
  \begin{pmatrix}
    1 \\
    1 \\
    1
  \end{pmatrix}
  =
  \begin{pmatrix}
    1-\frac{d_{11}}{d_1} \\
    1-\frac{d_{22}}{d_2} \\
    1-\frac{d_{33}}{d_3}
  \end{pmatrix}
  \ .
\end{equation} 

From \Cref{lem:d11-d1-fraction-Nc-restriction} we know that
$\frac{d_{11}}{d_1}\neq1$ for all $N>1$, so, in particular, also for
$N=3$. Therefore, the right-hand side of $E(1)$ is nonzero, which means that
\begin{equation}
  \label{eq:2}
  A_{12} \neq A_{13}
  \qquad \text{and} \qquad
  \text{not both $A_{12}$ and $A_{13}$ are zero}.
\end{equation} 
We distinguish two cases:
\begin{enumerate}
\item If $A_{12}\neq0$ and $A_{13}\neq0$, then both $\chi_{12}$ and
  $\chi_{13}$ can be determined uniquely from $E(1)$ by
  \Cref{lem:sum-nonzero-subsum}.
  \begin{enumerate}
  \item If $A_{23}\neq0$, we may
    also uniquely determine $\chi_{23}$.
  \item If $A_{23}=0$ for $N=3$, the corresponding $6j$~symbols
    $S_{22}^{23}$ and $S_{33}^{23}$ both vanish and hence $\chi_{23}$
    is irrelevant.
  \end{enumerate}
\item If only one of $A_{12}$ and $A_{13}$ is nonzero (i.e.\ only one
  of $\chi_{12}$ and $\chi_{13}$ can be determined uniquely), this
  means that the other $6j$~symbol is zero, making the corresponding
  $\chi_{ij}$ irrelevant. Without loss of generality, suppose that
  $\chi_{12}$ is uniquely determinable and hence
  $S_{11}^{13}=0=S_{33}^{13}$ (the analogous argument can be made if
  only $\chi_{13}$ is uniquely determinable).
  \begin{enumerate}
  \item If $A_{23}\neq0$, we may
    also uniquely determine $\chi_{23}$ from $E(2)$ using our result
    for $\chi_{12}$.
  \item If $A_{23}=0$, the corresponding $6j$~symbols $S_{22}^{23}$
    and $S_{33}^{23}$ both vanish,
    making $\chi_{23}$ irrelevant.
  \end{enumerate}
\end{enumerate}
Therefore, for $N=3$ all signs of the nonzero $6j$~symbols are
uniquely determinable.

Let us briefly comment on possibly non-existing $6j$~symbols: Notice
that, for a particular diagram $\alpha$, boxes may be added only to
the second row but not the third row (or vice versa), implying that
the $6j$~symbols with an index $3$ (resp.~$2$) do not exist. Examples
of these are
\begin{subequations}
  \begin{IEEEeqnarray}{rCl}
    \alpha = \ydiagram{3}
    & \quad \longrightarrow \quad &
    \text{boxes can only be added in rows $1$ and $2$} \\[3mm]
    \alpha = \ydiagram {2,2}
    & \quad \longrightarrow \quad &
    \text{boxes can only be added in rows $1$ and $3$}
    \ .
  \end{IEEEeqnarray} 
\end{subequations}
Then, in both cases, the matrix
equation~\eqref{eq:NC3-relations-6j-linear-matrix} reduces to
something even simpler, namely
\begin{equation}
  \label{eq:NC3-relations-6j-linear-matrix-vanishing-rows}
  \begin{pmatrix}
    0 & \chi_{1j}A_{1j} \\
    -\chi_{1j}A_{1j} & 0 
  \end{pmatrix}
  \begin{pmatrix}
    1 \\
    1
  \end{pmatrix}
  =
  \begin{pmatrix}
    1-\frac{d_{11}}{d_1} \\
    1-\frac{d_{jj}}{d_j} 
  \end{pmatrix}
  \quad
  \text{where $j=2$ (resp. $j=3$)}
  \ ,
\end{equation} 
and hence the signs of the only remaining $6j$~symbols $S^{1j}_{11}$
and $S^{1j}_{jj}$ can be determined
from the sign of $1-\frac{d_{11}}{d_1}\neq0$.

\paragraph{Going beyond~$N=3$:} Notice that for $N=3$, each equation $E(n)$ has
exactly $N-1=2$ terms on the left-hand side. Since the
right-hand side of $E(1)$ is non-zero for all values of $N$, this, in particular,
allows us to uniquely determine the sign of both terms of $E(1)$ for
$N=3$, and thus also for one of the two terms appearing on the
left-hand sides of $E(2)$ and $E(3)$, respectively. This is no longer
the case for $N>3$ as the left hand-sides of equations $E(n)$ contain
$N-1>2$ terms, and further information is needed to uniquely determine
their signs.

\section{Vertex properties}\label{app:properties-vertices}

We discuss the behavior of vertices under line swapping. 
Consider three irreps $\alpha$, $\beta$ and $\gamma$ with
$\gamma^* \subset \alpha \otimes \beta$. For each instance of
$\gamma^*$ in $\alpha \otimes \beta$ we introduce --- for the sake of argument --- two
vertices
\begin{equation}
  \label{eq:two-vertices-with-same-three-irreps}
  \FPic{GenRep-Vertex-CircleBlack-gammaO-alphaO-betaO}
  \hspace{1cm} \text{and} \hspace{1cm} 
  \FPic{GenRep-Vertex-CircleWhite-gammaO-betaO-alphaO}
  \ ,
\end{equation} 
which differ by line ordering (and possibly other ``internal'' vertex structure).
If the multiplicity of $\gamma^*$
in $\alpha \otimes \beta$ is one (or, equivalently, if the multiplicity of $\beta^*$ in $\alpha \otimes \gamma$ is one or, equivalently, if the multiplicity of $\alpha^*$ in $\beta \otimes \gamma$ is one), then the two vertices must be proportional
\begin{equation}
  \label{eq:swap-lines-in-vertex}
  \FPic{GenRep-Vertex-CircleBlack-gammaO-betaOSTAR-alphaOSTAR}
  \; = \;
  c \;
  \FPic{GenRep-Vertex-CircleWhite-gammaO-betaO-alphaO}
  \ ,
\end{equation} 
with some non-zero a priori complex constant $c$. All vertices
appearing in this work have multiplicity one, as can be seen
from Young diagram multiplication, using that in each vertex at least
one line is in the fundamental representation.

Next, we consider the complex conjugates of both vertices, and temporarily introduce different symbols for them,
\begin{equation}
  \left( \;
    \FPic{GenRep-Vertex-CircleBlack-gammaO-alphaO-betaO} \;
  \right)^*
  = \;
  \FPic{GenRep-Vertex-SquareBlack-gammaI-alphaI-betaI}
  \hspace{1cm} \text{and} \hspace{1cm} 
  \left( \;
  \FPic{GenRep-Vertex-CircleWhite-gammaO-betaO-alphaO} \;
  \right)^*
  = \;
  \FPic{GenRep-Vertex-SquareWhite-gammaI-betaI-alphaI}
  \ .
\end{equation} 
The complex conjugate of \Cref{eq:swap-lines-in-vertex} reads
\begin{equation}
  \label{eq:swap-lines-in-vertex-complex-conj}
  \FPic{GenRep-Vertex-SquareBlack-gammaI-betaISTAR-alphaISTAR}
  \; = \;
  c^* \;
  \FPic{GenRep-Vertex-SquareWhite-gammaI-betaI-alphaI}
  \ ,
\end{equation} 
and we can use these two equations in order to relate two $3j$~symbols, 
\begin{equation}
  \RPic{3j-alpha-beta-gamma--Vertex-CircleBlack-SStateT-Label}
  \RPic{3j-SState-ArrRev}
  \RPic{3j-alpha-beta-gamma--Vertex-SquareWhite-SState}
  \; \xlongequal{\text{\eqref{eq:swap-lines-in-vertex}}} \;
  c
  \;
  \RPic{3j-beta-alpha-gamma--Vertex-CircleWhite-SStateT-Label}
  \RPic{3j-s12-SState-ArrRev}
  \RPic{3j-alpha-beta-gamma--Vertex-SquareWhite-SState}
   \; \xlongequal{\text{\eqref{eq:swap-lines-in-vertex-complex-conj}}} \;
  \vert c \vert^2
  \;
  \RPic{3j-beta-alpha-gamma--Vertex-CircleWhite-SStateT-Label}
  \RPic{3j-SState-ArrRev}
  \RPic{3j-alpha-beta-gamma--Vertex-SquareBlack-SState}
   \ .
\end{equation} 
If we normalize all $3j$~symbols in the same way (we prefer to set them
to 1, but the argument also works for any other normalization) then we
conclude that $|c|=1$. In fact, if
$\alpha\neq\beta\neq\gamma\neq\alpha$, we can, and do, always choose
$c=1$, which is the most natural choice.

However, if two of the three irreps meeting in a vertex are
equivalent, say $\beta=\alpha$, then the two vertices defined in
\Cref{eq:two-vertices-with-same-three-irreps} are proportional to
each other, and we make the natural choice
\begin{equation}
  \FPic{GenRep-Vertex-CircleBlack-gammaO-alphaO-alphaO}
  \; = \; 
  \FPic{GenRep-Vertex-CircleWhite-gammaO-alphaO-alphaO}
  \ ,
\end{equation} 
any other choice would give a redundant definition.
In this case \Cref{eq:swap-lines-in-vertex} becomes
\begin{equation}
  \FPic{GenRep-Vertex-CircleBlack-gammaO-alphaOSTAR-alphaOSTAR}
  \; = \;
  c \;
  \FPic{GenRep-Vertex-CircleBlack-gammaO-alphaO-alphaO}
  \ ,
\end{equation} 
and by intertwining the upper two lines in the last equation we also obtain
\begin{equation}
  \FPic{GenRep-Vertex-CircleBlack-gammaO-alphaO-alphaO}
  \; = \;
  c \;
  \FPic{GenRep-Vertex-CircleBlack-gammaO-alphaOSTAR-alphaOSTAR}
  \ .
\end{equation} 
Finally, dividing both sides with $c$, $c=\frac{1}{c}$, and we find $c=\pm1$, i.e.\,
we are left with a sign.
In this work, the only irrep which can appear more than once in a
vertex is the fundamental representation, or, in other words, if two
identical lines meet in a vertex then they are always
quark-lines. Hence, there are only two vertices of this kind relevant
for this work, one with $c=1$ and one with $c=-1$,
\begin{equation}
  \label{eq:two-two-quark-vertices}
  \FPic{GenRep-Vertex-SymO-fundISTAR-fundISTAR}
  \; = \;
  \FPic{GenRep-Vertex-SymO-fundI-fundI}
  \hspace{1cm} \text{and} \hspace{1cm}
  \FPic{GenRep-Vertex-ASymO-fundISTAR-fundISTAR}
  \; = \;
  - \;
  \FPic{GenRep-Vertex-ASymO-fundI-fundI}
  \ ,
\end{equation} 
warranting the use of the same symbol $\bullet$ for all vertices. 

From here on, we again use the same symbol $\bullet$ for both
vertices~$\FPic{GenRep-Vertex-CircleBlack}$
and~$\FPic{GenRep-Vertex-SquareBlack}$~, an let the arrow direction
determine which vertex is intended. We also no longer use the
vertices~$\FPic{GenRep-Vertex-CircleWhite}$
and~$\FPic{GenRep-Vertex-SquareWhite}$ but instead swap lines on the
vertex $\bullet$. If an equation becomes more legible with two
lines swapped in a vertex, we indicate this swapping of lines by a
barred vertex (cf.\ \Cref{sec:line-ordering}), i.e.\ we define
\begin{equation}
  \FPic{GenRep-VertexSTAR-gammaO-betaO-alphaO}
  \; = \; 
  \FPic{GenRep-Vertex-gammaO-betaOSTAR-alphaOSTAR}
  \ ,
\end{equation} 
and for the purpose of this work we only have to keep in mind that
there is exactly one vertex, see
\Cref{eq:two-two-quark-vertices}, for which omitting a bar leads
to a sign change.

\section{Special cases for line ordering in vertices}
\label{app:special}

In \Cref{app:properties-vertices} we explained that we can largely
ignore the line ordering in (barred) vertices of the $6j$~symbols,
since for the $6j$~symbols we study most vertices connect three
\emph{distinct} irrep lines. 
The only exception are vertices with two incoming or two outgoing
quark-lines, and among these vertices only the vertex
\begin{equation}
  \label{eq:special-case-vertex}
  \FPic{GenRep-Vertex-ASymO-fundI-fundI}
\end{equation} 
is antisymmetric in the two quark-lines, see
\Cref{eq:two-two-quark-vertices}.

In order to have two incoming or two outgoing quark-lines in a vertex
within the $6j$~symbols under investigation,
\begin{equation}
  \FPic{6j-Mij-blue-Mi--greenSTAR-Mj-alphaSTAR} \, , 
\end{equation} 
$\alpha$ or $M_i$ or $M_j$ needs to be the fundamental representation,
i.e.\ a quark-line.

If $\alpha=\ydiagram{1}$ then $M_i$ and $M_j$ can be either
$\ydiagram{2}$ or $\ydiagram{1,1}$, and the only $6j$s of this kind
with at least one antisymmetric vertex are
\begin{equation}
  \label{eq:6j-with-ASym-2quark-vertex}
  \FPic{6j-YD111-fund-YD11--fundSTAR-YD11-fundSTAR}
  \ , \quad 
  \FPic{6j-YD21-fund-YD11--fundSTAR-YD11-fundSTAR}
  \ , \quad 
  \FPic{6j-YD21-fund-YD11--fundSTAR-YD2-fundSTAR}
  \quad \text{and} \quad 
  \FPic{6j-YD21-fund-YD2--fundSTAR-YD11-fundSTAR}
  \ ,
\end{equation} 
which have all been explicitly discussed and calculated in
Ref.~\citenum{Sjodahl:2018cca}.

If $M_i=\ydiagram{1}$, then $\bm{M}^{ij}$ can be either $\ydiagram{2}$
or $\ydiagram{1,1}$, from which it follows that also
$M_j=\ydiagram{1}$, and $\alpha$ can be either the trivial
representation (singlet) or the adjoint representation (a
gluon-line). If $\alpha$ is a singlet then the $6j$~symbol, up to
normalization, reduces to a $3j$~symbol. If $\alpha$ is the adjoint
representation, then the only $6j$ of this kind with at least one
antisymmetric vertex is
\begin{equation}
  \FPic{6j-YD11-fund-fund-fundSTAR-fund-adj} \ ,
\end{equation} 
which has also been calculated in Ref.~\citenum{Sjodahl:2018cca}; in
fact, it also reduces, up to normalization, to a $3j$ by the Fierz
identity.

\section{Dimensions of Young diagrams}\label{app:dimensions}

As is clear from \Cref{thm:6j-closed-form-expression}, calculating the
dimensions of the irreps contained in a $6j$ is imperative to
calculating the values of the $6j$~symbols discussed in this
paper. Therefore, we here recapitulate how to calculate these
dimensions directly from the corresponding diagrams.

First, we present the factors-over-hooks formula without
proof; proofs can be found in standard textbooks such
as Refs.~\citenum{Sagan:2000,Fulton:1997,Cvitanovic:2008zz}.

\ytableausetup{boxsize=5mm}
Consider a Young diagram $\lambda$. For each of its cells,
we may define a \emph{factor} and a \emph{hook length} in the
following way:
\begin{itemize}
\item the \emph{factor} $f_{a,b}$ of the cell $c_{a,b}\in\lambda$ in the
  $a^{\text{th}}$ row and the $b^{\text{th}}$ column is defined as
  \begin{equation}
    \label{eq:cell-factor-Def}
    f_{a,b} = b-a
    \ .
  \end{equation} 
\item the \emph{hook length} $h_{a,b}$ of the cell $c_{a,b}$ is defined to be
  the number of cells to the right of $c_{a,b}$ plus the number of cells
  below $c_{a,b}$ plus 1 ($c_{a,b}$ itself).
\end{itemize}
Then, the dimension of the $\SUN$ irrep corresponding to the diagram
$\lambda$ is given by
\begin{equation}
  \label{eq:factors-over-hooks}
  \text{dim}(\lambda)
  =
  \prod_{c_{a,b}\in\lambda}
  \frac{(N+f_{a,b})}{h_{a,b}}
  \ ,
\end{equation} 
where the product runs over all cells $c_{a,b}$ in $\lambda$.
Let us provide an example for illustration: Consider the Young diagram
\begin{equation}
  \label{eq:dim-lambda-Ex1}
  \lambda = \ydiagram{4,2,2,1}
  \ .
\end{equation} 
Then, the factors and hook lengths of each of the cells are
\begin{IEEEeqnarray}{rCCCl}
  & \text{factors:} &
  & \text{hook lengths:} &
  \nonumber \\
  &
  \begin{ytableau}
    0 & 1 & 2 & 3 \\
    -1 & 0 \\
    -2 & -1 \\
    -3
  \end{ytableau}
  & \hspace{2cm} &
  \begin{ytableau}
    7 & 5 & 2 & 1 \\
    4 & 2 \\
    3 & 1 \\
    1
  \end{ytableau}
  &
  \label{eq:dim-lambda-Ex2}
  \ ,
\end{IEEEeqnarray} 
such that the dimension of the irrep corresponding to
$\lambda$ is given by
\begin{equation}
  \label{eq:dim-lambda-Ex3}
  \text{dim}(\lambda)
  =
  \left[
    \frac{N}{7}
    \frac{(N+1)}{5}
    \frac{(N+2)}{2}
    (N+3)
  \right]
  \left[
    \frac{(N-1)}{4}
    \frac{N}{2}
  \right]
  \left[
    \frac{(N-2)}{3}
    (N-1)
    \right]
    \left[
      N-3
    \right]
  \ ,
\end{equation} 
which is zero for $N\leq3$ and, for example, becomes $36$ for $N=4$.
\ytableausetup{boxsize=3mm}

In QCD, a general Fock space sector may contain fundamental,
antifundamental and also adjoint factors. Young diagrams for irreps on
such a sector reflect this by containing the following conglomerates
of boxes,
\begin{IEEEeqnarray}{rCCCCCl}
  \label{eq:reps-Young-diag}
  & \text{fundamental:} &
  \hspace{1.5cm}
  & \text{antifundamental:} &
  \hspace{1.5cm}
  & \text{adjoint:} &
  \nonumber \\[2mm]
  & \FPic{YDiag-1} &
  & \FPic{YDiag-Nm1--NLabel} &
  & \FPic{YDiag-Nm1-1--NLabel} &
  \\[2mm]
  & \text{dim} = N &
  & \text{dim} = N &
  & \text{dim} = N^2-1 \ , &
  \nonumber
\end{IEEEeqnarray} 
where the dimensions can be verified using the factors-over-hooks
formula. For $\SUN$, a column of length $N$ may be crossed out in the
calculation of the dimension (as the hook lengths in this column will
cancel with the factors at the end of the respective rows), but it may
be preferable to not go the route of first adding boxes that will
ultimately be taken away. King~\cite{King:1970} offers such a way
in terms of back-to-back tableaux, where columns of length $N-1$
are represented as boxes that are added \emph{to the left} of the
given Young diagram. We will not review this method of multiplying
diagrams and calculating the corresponding dimensions here but rather
refer readers to the original source, Ref.~\citenum{King:1970}.

\nocite{*}
\bibliography{wigner6j}

\end{document}